\documentclass[12pt]{article}
\usepackage{amsfonts}
\usepackage{amsmath}
\usepackage{amsthm}
\usepackage{graphicx}
\newtheorem{thm}{Theorem}
\newtheorem{lemma}{Lemma}
\newtheorem{conj}{Conjecture}
\newtheorem{open}{Open question}

\title{Avalanche dynamics of the Abelian sandpile model on the expanded cactus graph}
\author{Gregory Gauthier\footnote{
The author gratefully acknowledges the support of Professor Mark Sapir of Vanderbilt University
in the research and writing of this paper, as well as support from National
Science Foundation Grant DMS\# 0700811 for research during the summer of 2011.  Additionally, the author thanks Tatiana Smirnova-Nagnibeda
and Michel Matter for their helpful discussions and suggestions,
as well as the referees for their review of this paper.}}
\begin{document}
\maketitle

\begin{abstract}
We investigate the avalanche dynamics of the Abelian sandpile model
on arbitrarily
large balls of the expanded cactus graph
(the Cayley graph of the free product $\mathbb{Z}_3 *
\mathbb{Z}_2$).
We follow the approach of Dhar and Majumdar%
\cite{DharMajumdar90}
to enumerate the number of recurrent configurations.
We also propose a method of enumerating
all the recurrent configurations in which adding
a grain to a designated origin vertex
(far enough away from the boundary vertices) causes topplings to 
occur in a specific cluster (a connected subgraph that is the union of cells, or copies of the 3-cycle) within the first wave of an avalanche.
This filling method lends itself to combinatorial evaluation of the number of recurrent configurations in
which a certain number of cells topple in the first wave of an avalanche
starting at the origin, which are
amenable to analysis
using well-known recurrences and corresponding generating functions.
Using this, we show that the cell-wise first-wave critical exponent of
the expanded cactus is $\frac{3}{2}$.
\end{abstract}

\section{Introduction}
\subsection{Definition of the Abelian sandpile model}
Here we introduce our definition of the Abelian sandpile model.\footnote{Good sources
for basic information on the Abelian sandpile model found in this section
are \cite{Dhar99} and \cite{MatterNagnibeda11}}

Let $(G_i)_{i=1}^\infty$ be an infinite sequence of finite graphs that is \emph{increasing}
(that is, each graph in the sequence is a subgraph of its successor), where
the sequence exhausts an
infinite $d$-regular transitive graph $\Gamma$
(that is, $\Gamma = \bigcup_{i=1}^\infty G_i$).
For example, $(G_i)_{i=1}^\infty$
might be a sequence of finite connected
graphs that each contain a designated origin vertex $o$ in $\Gamma$.\footnote{%
If $\Gamma$ is transitive, it does not matter which vertex is designated as the origin.}
Hereafter, we will consider $(G_i)_{i=1}^\infty$ to be an increasing sequence of
connected induced finite
subgraphs of $\Gamma$ containing the origin vertex.
If we define $G_i = B_{\Gamma}(o, i) = \{v \in V(\Gamma) | d_{\Gamma}(o, v) \leq i\}$
(the balls centered at $o$ of radius $i$ according to $d_{\Gamma}$, the graph distance function), then the
$G_i$'s are an increasing sequence of graphs that exhaust $\Gamma$.

The Abelian sandpile model on each graph $G_i$ in the sequence consists of
configurations, which are functions $\chi: V(G_i) \to \mathbb{Z}^+$ that
assign each vertex a \emph{height} that is a positive integer number of grains stacked on that vertex.  A \emph{stable} configuration is a configuration where each vertex has
height less than or equal to $d$.
Transitions from one stable configuration $\chi$ to another are governed by the following rules:
\begin{enumerate}
\item A grain of sand is added to a randomly selected vertex $v$ of $G$.
\item If the addition of that grain of sand causes the vertex $v$ to have a height exceeding $d$,
then $v$ \emph{topples}, causing it to lose $d$ grains and causing
each neighbor of $v$ in $G$ to gain one grain.  (There may be fewer than $d$ neighbors of $v$ in $G$,
in which case the total number of grains in the system will decrease.)\footnote{
Matter and Nagnibeda \cite{MatterNagnibeda11} alternatively consider only finite graphs $G$
whose vertices can be classified as dissipative; dissipative vertices do not topple.  Presumably
the remaining non-dissipative vertices each have degree $d$.  Given our notion of the Abelian sandpile
model, we can add the vertices in $V(\Gamma) - V(G)$ adjacent to a vertex in $G$ as dissipative vertices;
the model thus generated is clearly equivalent.}
\item Vertices gaining grains from toppling can also topple themselves if their height now exceeds $d$,
causing a chain reaction of topplings.%
\footnote{In the Abelian sandpile model, the order in which vertices topple does not
matter.  This is also true in more generalized sandpile models; see \cite{Dhar99}.}
\item Eventually, the toppling must stop in a configuration where each vertex's height is less than or equal to
$d$.  (This is true since there is at least one vertex
in $G$ whose degree is less than $d$; see \cite{Dhar99}.)
The process of toppling vertices iteratively in a configuration to
reach a stable configuration is known
as \emph{relaxation}.
At this point, another grain of sand is added to a
randomly selected vertex of $G$.
\end{enumerate}
If, when we start with a configuration $\chi$ of $G$,
adding a grain to a vertex $v$ of $G$ causes an avalanche,
then we can consider relaxing it by toppling a sequence of
not necessarily distinct vertices $v_1, \ldots, v_k$.
We denote by $\chi_0$ the configuration $\chi$ with
an extra grain at $v$ and $\chi_i$ the configuration reached
from $\chi_{i-1}$ by toppling $v_i$.  In order for this
to be well-defined, we need $\chi_{i-1}(v_i) > d$ for all
$i$ from $1$ to $k$.  We thus define the sequence $v_1, \ldots, v_k$
to be a
\emph{toppling sequence} if it satisfies the following conditions:
\begin{enumerate}
\item For all integers $i$ from $1$ to $k$, $\chi_{i-1}(v_i) > d$, and
\item The configuration reached after adding a grain to $v$
and toppling $v_1, \ldots, v_k$ in order is stable.
\end{enumerate}

Note that, in the sandpile model, unstable configurations never appear outside of the
process of adding a grain and relaxing.  Therefore, from this point forward, all
configurations are assumed to be stable.

\subsection{The sandpile model as a Markov process}
We can naturally define for each $v \in V(G)$ an
operator $a_v$ that acts on configurations by adding a grain to $v$ and relaxing until
a valid configuration is reached.\cite{Dhar99}  If we look only at recurrent configurations
(described below),
then we can also define inverse operators $a_v^{-1}$ on configurations $\chi$
as the unique \emph{recurrent} configuration $\chi'$
that is changed to $\chi$ under $a_v$.  The operators and their inverses together generate
a finite Abelian group.

If we assign to each $v$ a positive probability $p_v$, then application of $a_v$'s
successively to recurrent configurations according to those probabilities is a Markov process.%

\subsection{Avalanches, recurrent and forbidden configurations, and critical exponents}

The series of topplings caused by the addition of a grain to a vertex is called an \emph{avalanche},
and its
\emph{mass} equals the number of vertices that topple at least once due to the added grain.

Within a configuration, a \emph{forbidden subconfiguration} (FSC)
is a set $V' \subseteq V$  with the property
that every vertex $v \in V'$ has a height less than or equal to the number of edges incident to $v$
and a vertex of $V'$ (see
\cite{DharMajumdar90}).  For example, two adjacent vertices with height 1 form an FSC, as do three
vertices in a chain with heights 1--2--1.\cite{DharMajumdar90}
A configuration is \emph{recurrent} if it almost certainly occurs infinitely often in the Markov process.\cite{MeesterRedigZnamenski01}
We will use the well-known characterization of recurrent configurations.
\begin{lemma}[\cite{MeesterRedigZnamenski01}]
A configuration is recurrent if and only if it does
not have a forbidden subconfiguration.
\end{lemma}
In fact, it is
possible to equivalently define a configuration $\chi$ as recurrent if it is almost certain
that $\chi$ will appear infinitely many times in the Markov process.\cite{MeesterRedigZnamenski01}%
\footnote{Although the Markov process allows sand to be added to any vertex in a given finite
graph $G_i$ in the sequence of balls of $\Gamma$ centered at $o$, it is valid to only
examine what happens when a grain is added to $o$.  To examine what happens when a grain
is added to a different vertex $v$, a different graph sequence of balls centered at $v$
will need to be used.  But if $\Gamma$ is transitive (as the expanded cactus examined in
this paper is), the results will be the same regardless of which vertex is chosen as the
origin.  Although it is possible to consider the avalanche dynamics for every vertex
of a \emph{fixed} finite graph $G$ (which requires taking into account boundary effects),
this is beyond the scope of this paper.}

In the Abelian sandpile model with sinks (or non-dissipative vertices), one can use
the burning algorithm to determine whether a configuration $\chi$ is recurrent
or contains a forbidden subconfiguration.\cite{Dhar99}  At the start, only the sinks are burned.
At each step, the burning algorithm
examines each vertex $v \in G$ and determines whether $\chi(v)$ is greater than the number
of unburnt neighbors---if so, $v$ is burned.  This continues until every vertex is burned,
indicating that $\chi$ is a recurrent configuration, or until it is impossible to burn
any more vertices (in which case the unburnt vertices form a forbidden subconfiguration).
\cite{Dhar99}  The burning algorithm also facilitates a bijection between recurrent
configurations and spanning forests where each component contains exactly one sink.
Vertices burned on the first step must have fewer than $d$ unburnt neighbors and therefore
are adjacent to sinks,
while for a vertex $v$ to burn on the $t$th step (with $t > 1$)
requires that $v$ not be adjacent to a sink and that a neighbor of $v$ burn on the $t-1$st
step.  In a recurrent configuration, every vertex burns, and so one can construct a forest
where each component contains exactly one sink
by connecting each vertex that burns on the first step to a neighboring sink and each
vertex that burns on the $t$th step ($t > 1$) to a neighbor that burns on the $t-1$st step.
Although there may be more than one possible choice of neighbor, with appropriate conventions
on which neighbor to select based on the height of the vertex, this construction serves as a
bijection between recurrent configurations and spanning forests where each component contains
exactly one sink.\cite{Dhar99}

It is well-known that, in the steady state,
each recurrent configuration is equally
probable.
For each $k \in \mathbb{Z}^+$, $n \in \mathbb{Z}^+ \cup \{0\}$,
we can calculate both the number of recurrent configurations of $G_k$ as well as the number
of those recurrent configurations in which adding a grain to $o$ causes exactly $n$ distinct vertices to topple.
Define $N_k$ to be the number of recurrent configurations of $G_k$ and $N_{k, i}$ to
be the number of recurrent configurations of $G_k$ in which adding a grain to $o$ causes
exactly $i$ distinct vertices to topple.  Then define
$$p_k(i) = \frac{N_{k, i}}{N_k}$$
As $k$ goes to infinity, the values $p_k(n)$ should converge for each fixed $n$.  Define
$$p(n) = \lim_{k \to \infty}p_k(n)$$
Then we can consider the asymptotic behavior of the sequence $\{p(n)\}_{n=0}^\infty$.
In particular,
if there exist constants $C_1, C_2>0$, $N \in \mathbb{Z}^+$, and $\alpha \in \mathbb{R}$
such that $C_1 n^{-\alpha} < p(n) < C_2 n^{-\alpha}$, for all $n>N$ where
$p(n) > 0$, then $\alpha$ is the
\emph{critical exponent} of the increasing graph sequence
$(G_i)_{i=1}^\infty$ exhausting $\Gamma$.

As we will see, the problem becomes more amenable when we consider only those vertices
that topple on the first wave of an avalanche.  A vertex $v$ \emph{topples in the first
wave} of an avalanche caused by adding a grain at $o$ if there exists a toppling sequence
containing $v$ in which the vertices in the sequence up to and including
the first occurrence of $v$ include $o$ exactly once.
Define $N_{k, i}^f$ to be the number of recurrent configurations
of $G_k$ in which adding a grain to $o$ causes exactly $i$ distinct vertices to topple
in the first wave.  Then $p_k^f(i)$ and $p^f(n)$ are defined analogously to
$p_k(i)$ and $p(n)$.  This then allows us to state the
first-wave critical exponent $\alpha^f$ as a number satisfying that
$C_1 n^{-\alpha} < p^f(n) < C_2 n^{-\alpha}$ for all $n > N$ where $p^f(n) > 0$
for appropriate constants $C_1, C_2 > 0$.

In this paper, we examine the Abelian sandpile model on arbitrarily large but finite subgraphs of
the \emph{expanded cactus}
(as named by \cite{FisherEssam61}), which is formed by taking the 3-regular Bethe
lattice and ``decorating'' each vertex by replacing it with a three-vertex cycle, and the three edges
incident to that vertex in the Bethe lattice each become incident to a separate vertex in the three-vertex
cycle.  (See Figure \ref{fig:expanded-cactus}.)

\begin{figure}[htbp]
\begin{center}
\includegraphics[height=0.3\textheight]{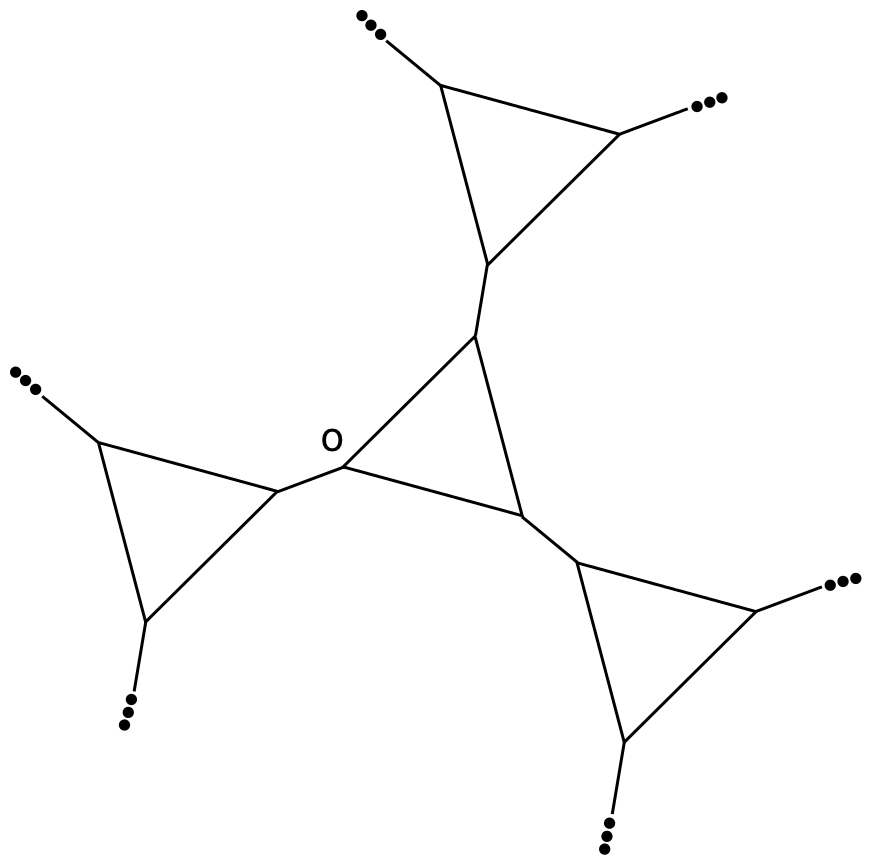}
\end{center}
\caption{The infinite expanded cactus with the origin vertex labeled.}
\label{fig:expanded-cactus}
\end{figure}

Each copy of the three-vertex cycle in the expanded cactus is designated a \emph{cell}.
We will only consider finite expanded cacti that can be generated by decorating a finite subgraph
of the Bethe lattice (the \emph{underlying graph}).
The expanded cactus is
the Cayley graph of the group $\mathbb{Z}_3 * \mathbb{Z}_2$
with respect to the standard generating set.  The maximum height of each
vertex is 3.  In the infinite expanded cactus $\Gamma$,
we will take the graph
sequence $(G_i)$ defined as follows: take the infinite 3-regular Bethe
lattice $H$ and fix an origin vertex $z$ in $H$.  $G_i$ will be given
by the decoration of $B_{H}(z, i)$, the ball of radius $i$ about $z$ in
the Bethe lattice.\footnote{Matter and Nagnibeda \cite{MatterNagnibeda11}
examined the critical exponent for cactus graphs in the \emph{random
weak limit}.  Their limiting process involves taking larger and
larger balls that exhaust the graph,
but the origin is a randomly selected vertex within each
ball.  Therefore, in the limit, the origin is not a fixed vertex
arbitrarily deep within the graph, but rather has its origin located
according to a probability distribution such that the origin is a
finite distance
away from a vertex with degree less than $d$ with probability $1$.
Thus, the graphs studied by Matter and Nagnibeda are one-ended,
which is used in their analysis.  However, the graphs studied in
this paper are infinitely-ended, so the approach used in
\cite{MatterNagnibeda11} is not applicable to our problem.}
The cell formed by decorating $z$, present in
each $G_i$, will be designated the origin cell; additionally, one
arbitrary vertex in that origin cell will be designated the origin vertex.

Although calculating the distribution of avalanches of specific mass requires calculating
the probability that adding a grain will cause a given number of vertices to topple, it is simpler,
in the case of the expanded cactus, to consider a different definition of avalanche mass:
the number of \emph{cells} in which at least one vertex topples (which we will call
the \emph{cell mass} of the avalanche).
To examine the behavior of the Abelian sandpile model on
$\Gamma$, we will consider a graph sequence $(G_i)_{i=1}^\infty$ of the expanded cactus $\Gamma$,
where $G_i$ is the decoration of the ball centered at the origin of the underlying 3-regular tree
with radius $i$.  Analogously to how $p_k(n)$ and $p(n)$ are defined in general, we will define
$p_k^c(n)$ and $p^c(n)$ using the number of recurrent configurations where exactly $n$ \emph{cells}
topple when a grain is added to the origin vertex in place of the number of
recurrent configurations where exactly
$n$ vertices topple when a grain is added to the origin vertex.
Likewise, we say that a cell is toppled on the first wave if any of its vertices
is toppled in the first wave; then, we define $p_k^{cf}(n)$ and $p^{cf}(n)$ in
a similar fashion using the cells that topple in the first wave.

If there exist constants $C_1, C_2>0$, $N \in \mathbb{Z}^+$, and $\alpha \in \mathbb{R}$
such that $C_1 n^{-\alpha} < p^{fc}(n) < C_2 n^{-\alpha}$, for all $n>N$ where
$p(n) > 0$, then $\alpha$ is the
\emph{cell-wise first-wave critical exponent} of the increasing graph sequence
$(G_i)_{i=1}^\infty$ exhausting $\Gamma$.
If the critical exponent,
(respectively, first-wave critical exponent%
, cell-wise first-wave critical exponent)
is $\alpha$ for every increasing graph sequence $(G_i)_{i=1}^\infty$
exhausting $\Gamma$, then we say that the critical exponent
(respectively, first-wave critical exponent%
, cell-wise first-wave critical exponent)
of $\Gamma$ is $\alpha$.

The critical exponent
of increasing graph sequences $(G_i)_{i=1}^\infty$ exhausting the expanded cactus
in the Abelian sandpile model
is important to mean-field theory, which concerns itself with the size of
avalanches in critical systems.\cite{ZapperiLauritsenStanley95}  Dhar and Majumdar \cite{
DharMajumdar90}
showed that the critical mean-field exponent for
infinite binary (3-regular) trees is $\frac{3}{2}$.  Although it is
conjectured
that the critical exponents of quasi-isometric graphs should be the same, this has
not been proven for non-trivial pairs of infinite graphs.  Our result,
below, provides evidence consistent with this conjecture:
\begin{thm}
On the expanded cactus, the cell-wise
first-wave critical exponent is $\frac{3}{2}$.
\end{thm}

In percolation
theory, there is a similar proposition that the critical exponents
of percolation on a transitive graph
only depend on the large-scale structure of the graph
\cite{Kozakova08} and therefore, quasi-isometric graphs will have the same
critical exponents for percolation.  However, there are few
calculated examples that support this conjecture.

\subsection{Outline of method of Dhar and Majumdar and our approach}
Dhar and Majumdar \cite{DharMajumdar90} determine $p_i$ by counting,
for each sufficiently large subgraph $T$ of the $3$-regular tree,
the number of recurrent configurations (which appears in the denominator)
and the number of recurrent configurations where exactly $i$
vertices topple (which appears in the numerator).  Using the notions
of strongly and weakly allowed configurations on subtrees,
Dhar and Majumdar first compute the limiting ratio of strongly allowed
to weakly allowed configurations as the size of a subtree goes to infinity.
They then derive for every connected induced subgraph (\emph{cluster})
$C$ with $n$ vertices
a recurrence relation for the number of
recurrent configurations on $T$ in terms of the number of strongly
allowed configurations of each of the $n+2$ subtrees $U_1, \ldots, U_{n+2}$
induced by removing $C$.
In particular, Dhar and Majumdar showed that the number of recurrent
configurations is $3 \cdot 4^n \prod_{j=1}^{n+2} N_s(U_j)$, where
$N_s(U_j)$ denotes the number of strongly allowed configurations on $U_j$.
Then, Dhar and Majumdar show that a configuration $\chi$ on $T$ is recurrent
and has the property that adding a grain to the origin causes exactly
the $n$ vertices in a cluster $C$ to topple if and only if
\begin{enumerate}
\item $\chi(v) = 3$ for every $v \in C$
\item Each of the $n+2$ induced subtrees $U_1, \ldots, U_{n+2}$
has a root vertex whose height is $1$ or $2$
\item At least one of the $n+2$ induced subtrees is strongly allowed
\end{enumerate}
Dhar and Majumdar show that for each such cluster $C$ of size $n$,
the number of recurrent configurations where exactly the vertices in $C$
topple is $\left(1 - 2^{-(n+2)}\right) \prod_{j=1}^{n+2} N_s(U_j)$,
and this is independent of the shape of $C$.  Dhar and Majumdar
then derive the formula for the number of $n$-vertex clusters containing
the origin using generating functions, and they showed that the number of
clusters grows asymptotically to $4^n n^{-\frac{3}{2}}$.  From this,
it follows that
$$p_i = \left(1 - 2^{-(n+2)}\right) n^{-\frac{3}{2}}$$
from which it is clear that the critical exponent of the $3$-regular tree is
$\frac{3}{2}$.

The remainder of our paper analyzes the cell-wise
first-wave critical exponent of the expanded
cactus using a refinement of the method of Dhar and Majumdar.
In section 2, we describe how configurations can be enumerated
by dividing the finite graph into decorated rooted subtrees and defining
and considering weakly and strongly allowed radicals (configurations
restricted to decorated rooted subtrees) on those
decorated subtrees.  Based on this method, we can express the number
of recurrent configurations on a finite graph with regards to a cluster
(defined below) of cells about the origin in terms of the number of strongly
allowed radicals on the decorated subtrees induced by removing the
cluster.  Section 4 outlines how the number of recurrent configurations on
a rooted expanded cactus where
cells in a given cluster topple and no others can be enumerated in terms
strongly and weakly allowed radicals induced by removing the cluster,
using the concept of filling rules to determine what configurations
are allowed within the cluster.  Section 5 examines how we can adapt
the filling rules to enumerate configurations on the expanded cactus
where given cells topple on the first wave, and it examines issues
with multiple-wave avalanches on
the expanded cactus, issues arising from our graph's being a cactus
and not a tree as in \cite{DharMajumdar90}.
In section 6, we use asymptotic enumeration
methods to express the number of configurations where $n$ cells topple
in terms of the number of strongly allowed radicals induced by any cluster
of size $n$.  Our method of counting overcounts in that it includes
configurations that are not recurrent---section 7 shows that only a small
number of positions are erroneously counted and thus that the asymptotics
are not affected.
Finally, section 8 provides open questions and
a framework for generalizing
the method used in this paper to arbitrary decorations of infinite trees
by finite transitive connected graphs.

\section{Counting method---weakly and strongly allowed radicals}

Throughout this paper, our method of enumerating recurrent configurations
on the finite expanded cactus $G$
with certain properties is to take a \emph{cluster}
(or connected subgraph formed by the union of cells)
$C$ with $n$ cells, containing a predetermined cell called the origin,
one of whose vertices is also designated the origin
vertex.  The subgraph induced by $V \setminus C$ has $n + 2$ connected components, which are
\emph{decorated rooted subtrees}.
Each decorated rooted subtree has a cell
immediately adjacent to the cluster, which we will call the \emph{root cell}, and a
\emph{root vertex within}
the root cell that is adjacent to a vertex within the cluster.

Given a configuration $\chi$ and a decorated rooted subtree $U$, the restriction $\left.\chi\right|_U$
is called a \emph{radical}.  The radical's
root cell and vertex are the same as the decorated rooted subtree on which it is based.  Considering
the radical as not attached to anything, the configuration may be a recurrent configuration
 or have an FSC.
If the radical does not have a forbidden subconfiguration, then we will create a dichotomy,
analogous to the one introduced in \cite{DharMajumdar90}, of strongly allowed and weakly allowed
radicals.  A radical that does not contain an FSC is \emph{strongly allowed}
 if,
when attached to a vertex of height 1, does not result in an FSC, and is \emph{weakly allowed} otherwise.
Also, by definition, if $U$ is empty, then the null configuration on $U$
is a strongly-allowed radical.
Given a decorated rooted subtree $U$, we denote the number of weakly allowed radicals on $U$ as
$N_w(U)$, the number of strongly allowed radicals on $U$ as $N_s(U)$, and the ratio $\frac{N_w(U)}
{N_s(U)}$ as $x_U$.

In our enumerations, our goal is to express the number of recurrent configurations,
 as well as the number of
recurrent configurations
wherein topplings occur in exactly the cells of a given cluster $C$ of $n$ cells, in terms
of the product
$$\prod_{i=1}^{n+2} N_s(U_i)$$
where each $U_i$ is a decorated rooted subtree induced by $C$.  We first turn to counting the number
of recurrent configurations.

\section{Enumerating the recurrent configurations}
\label{sec:enumerating-recurrent-configurations}

Consider the origin cell in the expanded cactus and a configuration of both the cell and the
three induced radicals.  Whether the configuration is recurrent or not
can be determined entirely by
the heights of the vertices in the origin cell, as well as whether each radical is strongly or weakly
allowed.

In Table \ref{tab:recurrent-configuration-origin}, we examine the 16 possible cell configurations and determine which combinations
of strongly and weakly allowed radicals result in an recurrent configuration.
Here, S denotes a strongly allowed
radical, W a weakly allowed radical.  The order of the letters in each combination indicates what type
of radical is attached where, corresponding to the order of the vertices in the given cell.  For
example, the entry WSS in the 2-2-3 row indicates that if a weakly allowed radical is attached to the first
2 and strongly allowed radicals to the second 2 and to the 3, the resulting configuration is
recurrent.

\begin{table}[htbp]
\begin{center}
\begin{tabular}{cl}
Cell & recurrent radical combinations \\ \hline
1-2-3 & SSS \\
1-3-2 & SSS \\
2-1-3 & SSS \\
2-3-1 & SSS \\
3-1-2 & SSS \\
3-2-1 & SSS \\ \hline
2-2-3 & SSS, WSS, SWS \\
2-3-2 & SSS, WSS, SSW \\
3-2-2 & SSS, SWS, SSW \\ \hline
1-3-3 & SSS, SWS, SSW \\
3-1-3 & SSS, WSS, SSW \\
3-3-1 & SSS, WSS, SWS \\ \hline
2-3-3 & SSS, WSS, SWS, SSW, WWS, WSW \\
3-2-3 & SSS, WSS, SWS, SSW, WWS, SWW \\
3-3-2 & SSS, WSS, SWS, SSW, WSW, SWW \\ \hline
3-3-3 & SSS, WSS, SWS, SSW, WWS, WSW, SWW \\ \hline
\end{tabular}
\end{center}
\caption{The combinations of allowed radicals for each allowed origin cell
configuration that result in a recurrent configuration.}
\label{tab:recurrent-configuration-origin}
\end{table}

We also must consider how to compute the number of strongly and weakly allowed radicals on a given
decorated rooted subtree.  This can be calculated inductively by considering the 16 possible configurations
of the root cell of a radical, as well as whether the two induced subradicals are weakly or strongly
allowed.  Each combination of a root cell configuration and weakly or strongly allowed subradicals produces
either a radical with a forbidden subconfiguration, a weakly allowed whole radical, or a strongly allowed
whole radical.

In Table \ref{tab:recurrent-configuration-radical}, we illustrate the combinations of root cells and weakly and strongly allowed radicals
that produce weakly and strongly allowed whole radicals.  The root vertex is the last number listed
in the cell; the two letter codes indicate whether a weakly allowed (W) or strongly allowed (S)
radical is being added to the vertex indicated by the first and second number, respectively.

\begin{table}[htbp]
\begin{center}
\begin{tabular}{cll}
Cell & Weakly allowed whole & Strongly allowed whole \\ \hline
1-2-3 & SS & --- \\
1-3-2 & SS & --- \\
2-1-3 & SS & --- \\
2-3-1 & SS & --- \\
3-1-2 & SS & --- \\
3-2-1 & SS & --- \\ \hline
2-2-3 & SS, WS, SW & --- \\
2-3-2 & WS & SS \\
3-2-2 & SW & SS \\ \hline
1-3-3 & SW & SS \\
3-1-3 & WS & SS \\
3-3-1 & SS, WS, SW & --- \\ \hline
2-3-3 & SW, WW & SS, WS \\
3-2-3 & WS, WW & SS, SW \\
3-3-2 & --- & SS, WS, SW \\ \hline
3-3-3 & WW & SS, WS, SW \\ \hline
\end{tabular}
\end{center}
\caption{The combinations of allowed radicals for each allowed root cell
configuration that result in strongly allowed or weakly allowed
whole radicals.}
\label{tab:recurrent-configuration-radical}
\end{table}

From Table \ref{tab:recurrent-configuration-radical}, we can calculate the limiting ratio of weakly allowed to strongly allowed
radicals on arbitrarily large decorated rooted subtrees.  To do this, we first tabulate
the number of instances of WW, WS, SW, and SS in the weakly allowed and the strongly allowed columns:

\begin{table}[htbp]
\begin{center}
\begin{tabular}{r|cccc}
 & WW & WS & SW & SS \\ \hline
Weakly allowed whole & 3 & 5 & 5 & 8 \\
Strongly allowed whole & 0 & 3 & 3 & 8
\end{tabular}
\end{center}
\caption{The number of root cell configurations for which the
given combinations of allowed subradicals produce a weakly
allowed or strongly allowed whole radical.}
\label{tab:weakly-strongly-allowed-radical-count}
\end{table}

Suppose $U$ is a decorated rooted subtree such that when the root cell is removed,
two smaller decorated rooted subtrees, $U_1$ and $U_2$ are induced.  From Table
\ref{tab:weakly-strongly-allowed-radical-count}, we know
\begin{eqnarray*}
N_w(U) & = & 3N_w(U_1)N_w(U_2) + 5N_w(U_1)N_s(U_2)\\
 & + & 5N_s(U_1)N_w(U_2) + 8N_s(U_1)N_s(U_2) \\
 & = & (3x_{U_1}x_{U_2} + 5x_{U_1} + 5x_{U_2} + 8)N_s(U_1)N_s(U_2)
\end{eqnarray*}
and
\begin{eqnarray*}
N_s(U) & = & 3N_w(U_1)N_s(U_2) + 3N_s(U_1)N_w(U_2) + 8N_s(U_1)N_s(U_2) \\
 & = & (3x_{U_1} + 3x_{U_2} + 8)N_s(U_1)N_s(U_2)
\end{eqnarray*}

Therefore,
$$x_U = \frac{3x_{U_1}x_{U_2} + 5x_{U_1} + 5x_{U_2} + 8}{3x_{U_1} + 3x_{U_2} + 8}$$

Consider now the decorated rooted subtrees generated as follows: $B_0$ as a single cell,
$B_n$ as a decorated rooted subtree whose child subtrees are $B_{n-1}$.  Let $x_n = \frac{N_w(T_n)}
{N_s(B_n)}$ for each nonnegative integer $n$.  Then
$$x_0 = 1$$
and
$$x_n = \frac{(8 + 10x_{n-1} + 3x_{n-1}^2)N_s(T_{n-1})^2}{(8 + 6x_{n-1}) N_s(T_{n-1})^2}
= \frac{8 + 10x_{n-1} + 3x_{n-1}^2}{8 + 6x_{n-1}}$$

This simplifies to
$$x_n = 1 + \frac{1}{2}x_{n-1}$$
so $x_n = 2 - 2^{-n}$, and it is clear that
$$\lim_{n \to \infty}x_n = 2$$

Further, we can compute $\frac{\partial x_U}{\partial x_{U_1}}$ as follows:
\begin{eqnarray*}
\frac{\partial x_U}{\partial x_{U_1}} & = & \frac{(3x_2+5)(3x_1 + 3x_2 + 8) - 3(3x_1x_2 + 5x_1 +
 5x_2 + 8)}{(3x_1 + 3x_2 + 8)^2} \\
  & = & \frac{9x_2^2+24x_2+16}{(3x_1 + 3x_2 + 8)^2}
\end{eqnarray*}
This means that $x_U$ increases as a function of $x_{U_1}$ and, by symmetry, $x_{U_2}$.
Thus, if $U$ is a decorated rooted subtree that contains a subtree $B_n$, then $x_U \geq x_n =
2 - 2^{-n}$.  Further, since $x_U(2, 2) = 2$, it is trivial to show inductively
that $x_U < 2$ for every finite decorated rooted subtree $U$.  Together, these results
show that we can guarantee that a decorated rooted subtree $U$ will have $x_U$ arbitrarily
close to $2$ simply by ensuring that $U$ is deep enough.

Further, we have the following result that relieves us from having to consider
which specific graph sequence we are using:
\begin{thm}
\label{thm:increasing-graph-sequence-epsilon}
Let $(G_n)_{n=1}^\infty$ be an increasing graph sequence that exhausts $\Gamma$,
the infinite expanded cactus, where each $G_n$ is a connected graph formed
by the union of cells, including the origin.  Then, for every $\epsilon > 0$,
and for every positive integer $k$, there exists a positive integer $N$ such
that for every cluster $C$ of at most $k$ cells in $G_N$ containing the origin cell
and every decorated rooted
subtree $U$ induced by $C$, it holds that $2 > x_U > 2 -\epsilon$.
\end{thm}

Now let us suppose that $C$ is a cluster of $n$ cells about the origin, inducing decorated rooted
subtrees $U_1, \ldots, U_{n+2}$.  We would like to express the number of recurrent configurations
on the entire
graph in terms of $N_s(U_1)\cdot N_s(U_2) \cdots N_s(U_{n+2})$.  The following analysis works with
clusters of any shape, but for simplicity, we will assume that the cluster is a chain of $n$ cells,
numbered consecutively from $1$ to $n$.  Label the two induced decorated rooted subtrees adjacent to cell
$1$ $U_1$ and $U_2$, the one decorated rooted subtree adjacent to cell $k$ (for $1 < k < n$) $U_{k+1}$,
and the two decorated rooted subtrees adjacent to cell $n$ $U_{n+1}$ and $U_{n+2}$.  Also, for $1 \leq k
< n$, define $T_k$ to be the decorated subtree whose root is cell $k$ and which does not contain cell
$k+1$.

With these labellings, we can compute the total number of recurrent configurations
on the graph by dividing the graph
at cell $n$ and using Table \ref{tab:recurrent-configuration-origin}.  The result is
$$N_{\textrm{recurrent}} = [16 + 8x_{T_{n-1}} + 8x_{U_{n+1}} + 8x_{U_{n+2}}$$
$$+ 3x_{T_{n-1}}x_{U_{n+1}} + 3x_{T_{n-1}}x_{U_{n+2}} +
3x_{U_{n+1}}x_{U_{n+2}}]\cdot$$
$$N_s(T_{n-1})N_s(U_{n+1})N_s(U_{n+2})$$

Further, we have the recurrence, for $2 \leq k < n$,
\begin{equation}
\label{eqn:first-recurrence}
N_s(T_k) = [8 + 3x_{U_{k+1}} + 3x_{T_{k-1}}]N_s(U_{k+1})N_s(T_{k-1})
\end{equation}
and
\begin{equation}
\label{eqn:second-recurrence}
N_s(T_1) = [8 + 3x_{U_1} + 3x_{U_2}]N_s(U_1)N_s(U_2)
\end{equation}

If $\epsilon > 0$ is given and $G$ is a finite graph satisfying the conclusion
of Theorem \ref{thm:increasing-graph-sequence-epsilon}, then we have
$$N_{\textrm{recurrent}} > [100 - 60\epsilon + 9\epsilon^2]N_s(T_{n-1})N_s(U_{n+1})N_s(U_{n+2})$$
and iterative expansion gives
$$N_{\textrm{recurrent}} > (20 - 6\epsilon)^{n-1} (100 - 60\epsilon + 9\epsilon^2)
\prod_{i=1}^{n+2}N_s(U_i)$$

Also, since $x_U < 2$ for every finite decorated rooted subtree $U$, we have
$$N_{\textrm{recurrent}} < 100\cdot 20^{n-1}\prod_{i=1}^{n+2}N_s(U_i)$$

Therefore, as $\epsilon$ approaches $0$, we have
$$N_{\textrm{recurrent}} = 5 \cdot 20^n \prod_{i=1}^{n+2} N_s(U_i)$$

\section{Enumerating the recurrent configurations on a decorated
\emph{rooted} subtree in which cells of a given cluster will topple using filling rules}
\label{sec:enumerating-the-configurations}

In this section, we propose filling rules that,
given a decorated rooted subtree and a cluster $C$ about the origin, describe what
combinations of configurations of the cells of $C$
and radicals induced by $C$ will result in an
recurrent configuration wherein adding a grain to the origin
causes topplings to occur within at least one
vertex in each cell of $C$ and no vertices outside the cells of $C$.
The rules
allow for independent assignment of configurations to each cell of $C$,
and requirements for the radicals
only depend on the cell to which the underlying decorated rooted subtree
is attached.  Therefore,
the number of recurrent configurations on a decorated rooted subtree
where a given cluster of cells topples can
be enumerated by simple use of the multiplication
principle of counting.

\subsection{Terminology}
To explain the filling rules clearly, we need to define some additional terms.
Within each cell, the \emph{origin-facing vertex} is the vertex
closest to the origin vertex.
Each cell $c$ within
$C$ is classified as internal, medial, or terminal, as follows.
If each non-origin-facing vertex of $c$ is attached to a
cell in the cluster, then $c$ is an \emph{internal cell}.
If exactly one non-origin-facing vertex of $c$ is attached to
a cell in the cluster, then $c$ is a \emph{medial cell}.
If no non-origin-facing vertices in $c$ are attached to
cells in the cluster, then $c$ is a \emph{terminal cell}.
Figure \ref{fig:sandpile-cactus-cluster} illustrates an example cluster
and the internal, medial, and terminal cells within it.

\begin{figure}[htbp]
\begin{center}
\includegraphics[height=0.3\textheight]{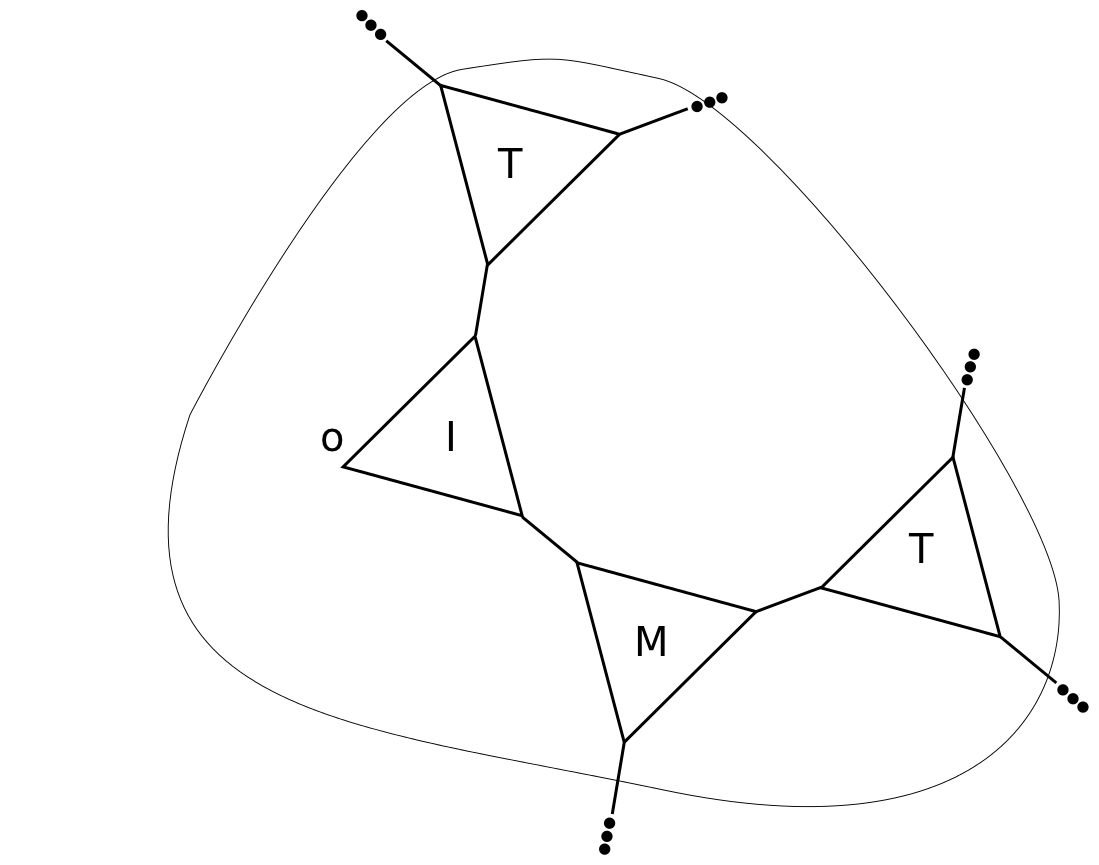}
\end{center}
\caption{An example cluster, with cells labeled as internal (I),
medial (M), or terminal (T).}
\label{fig:sandpile-cactus-cluster}
\end{figure}

Define a radical to be a \emph{stopper} if it is allowed and
its root vertex has height 1 or 2.  Also, we define the empty
radical to be a stopper.  For radicals that are sufficiently
deep, the following hold:
\begin{enumerate}
\item The number of stoppers equals the number of strongly allowed radicals.%
\footnote{This holds regardless of the shape of the radical, so long as it is a decorated rooted
subtree, since there is a bijection between stoppers and strongly allowed radicals.}
\item The number of strongly allowed stoppers is $\frac{7}{20}$ of the number of strongly allowed radicals.
\end{enumerate}
To see this, refer to Table \ref{tab:recurrent-configuration-radical}.
Since weakly allowed subradicals are twice as prevalent as strongly allowed
subradicals, weight each SS as 1, each SW and WS as 2, and each WW as 4.  Then calculating the total
weights for strongly allowed whole radicals, stoppers,
and strongly allowed stoppers, gives the preceding
results.

\subsection{The rules}
The recurrent configurations in which topplings occur in each cell in $C$ and no others are precisely
characterized by the following:
\begin{enumerate}
\item The cells in $C$ have configurations specified by the filling rules, according to their
being internal, medial, or terminal.
\item The radicals are all allowed, and they satisfy
the specific requirements for being strongly allowed
or a stopper that are
determined by the cell to which the underlying decorated rooted subtree is attached.
\end{enumerate}

\subsubsection{Internal cells}
The possible configurations for internal cells are
\begin{enumerate}
\item 3-3-3
\item 3-3-2, with the vertex closest to the origin having height 3 (this item encompasses two
cell configurations by symmetry)
\end{enumerate}

\subsubsection{Medial cells}
The possible configurations for medial cells, as well as the corresponding restrictions on
the attached radical, are
\begin{enumerate}
\item 3-3-3---radical must be a stopper
\item 3-3-1, with the vertex adjoining the induced decorated rooted subtree having height 1---radical
must be strongly allowed
\item 3-3-2, with the vertex adjoining the induced decorated rooted subtree having height 2---radical
must be a stopper
\item 3-3-2, with the vertex adjacent to another cell in $C$ but not the closest vertex to the origin
having height 2---radical must be a stopper
\end{enumerate}

\subsubsection{Terminal cells}
Finally, we turn to terminal cells, which have these possible configurations
that can be filled into the cell (and the corresponding
radical restrictions):
\begin{enumerate}
\item 3-3-3---both radicals must be stoppers
\item 3-3-2, with the vertex closest to the origin having height 3---both radicals must be stoppers
(this item encompasses two cell configurations)
\item 3-3-1, with the vertex closest to the origin having height 3---the radical on the decorated
rooted subtree adjacent to the 3 must be a stopper;
the other radical must be strongly allowed
(this item encompasses two cell configurations)
\item 3-2-2, with the vertex closest to the origin having height 3---at least one of the two radicals
must be strongly allowed; the other must be allowed
\item 3-2-1, with the vertex closest to the origin having height 3---both radicals must be strongly
allowed (this item encompasses two cell configurations)
\end{enumerate}

\subsection{Correctness of filling rules for decorated \emph{rooted}
subtrees}
In order to use the filling rules to count the number of recurrent
configurations wherein the cells in cluster $C$ topple and no others,
we prove the following result.
\begin{thm}
\label{thm:filling-rules}
Let $\chi$ be a configuration on $G$, a finite decorated rooted subtree,
and let $C$ be a cluster of $G$.
Then $\chi$ is a recurrent configuration and
adding a grain to the origin vertex causes exactly the cells in $C$ to
topple if and only if $\chi$ can be generated by applying the filling
rules to the cluster $C$ in $G$.
\end{thm}

To prove this, we will use the following propositions.
(Unless otherwise stated, $G$ is a connected finite subgraph with origin of an
infinite connected $d$-regular graph $\Gamma$.)
\begin{lemma}
\label{lem:origin-first-repeated-toppling-vertex}
Let $v \in G$ be given such that $\chi(v) = d$.  Let
$v_1, v_2, \ldots, v_k$ be a toppling sequence for
the avalanche caused by adding a grain to $v$.  Then,
if $v_1, v_2, \ldots, v_k$ are not all distinct, then the
first repeated vertex in the sequence is $v$.
\end{lemma}
\begin{proof}
Take $t$ to be the smallest positive
integer such that $v_t = v_{t'}$ for some $t' < t$.
For each $s$ from $1$ to $k$, let $\chi_s$
denote the not necessarily stable configuration obtained by starting
with $\chi$, adding a grain to $u$, and relaxing only the vertices
$v_1, \ldots, v_s$ in that order.
Now, we have that
$\chi_{t-1}(v_t) > d$ since $v_t$ is the $t$th toppling vertex in
the sequence.  But we also know that
$$\chi_{t-1}(v_t) = \chi(v_t) + 1(v_t = v) +
\sum_{s=1}^{t-1} 1(v_s \sim v_t) - d$$
However, the sum $\sum_{s=1}^{t-1} 1(v_s \sim v_t)$ is at most $d$,
since each of the $d$ neighbors of $v_t$ topples at most once
prior to the second toppling of $v_t$.  Thus, in order for
$\chi_{t-1}(v_t) > d$ to hold, we must have $v_t = v$.
\end{proof}

\begin{lemma}
\label{lem:origin-multiple-topplings}
Let $v \in G$ be given such that $\chi(v) = d$.
Then $v$ topples more than once in the avalanche
caused by adding a grain to $v$ if and only if
$v$ has $d$ neighbors in $G$ and each of those neighbors
topples.  In particular, if
$d_G(v) < d$, then no vertex topples more than once in
the avalanche caused by adding a grain to $v$.
\end{lemma}
\begin{proof}
Let $v_1, \ldots, v_k$ be a toppling sequence for the
avalanche caused by adding a grain to $\chi$ at $v$.

$(\Longrightarrow)$ Suppose $v$ topples more than once.
Let $t$ be the positive
integer such that $v_t = v$ and such that there exists
exactly one $t' < t$ such that $v_{t'} = v$.
Then, for each $s$ from $1$ to $k$, let $\chi_s$
denote the not necessarily stable configuration obtained by starting
with $\chi$, adding a grain to $u$, and relaxing only the vertices
$v_1, \ldots, v_s$ in that order.  We have
$$\chi_{t-1}(v) = d + 1 + \sum_{s=1}^{t-1} 1(v_s \sim v_t) - d$$
But since $v$ topples at step $t$ of relaxation, we have
$$\chi_{t-1}(v) \geq d + 1$$
This can only happen when $\sum_{s=1}^{t-1} 1(v_s \sim v_t)$
is its maximum value, $d$.  When
$\sum_{s=1}^{t-1} 1(v_s \sim v_t) = d$, we have that
$v$ has $d$ neighbors in $G$ and each of those neighbors topples,
as desired.

$(\Longleftarrow)$ Suppose that each neighbors of $v$ topples
more than once.  Choose $t$ to be the smallest positive integer
such that for each neighbor $u$ of $v$, there is a positive integer
$s \leq t$ such that $v_s = u$.  Then, by the converse,
$v$ has not toppled twice, so it has toppled exactly once by step $t$.
We have
$$\chi_t(v) = d + 1 + \sum_{s=1}^{t-1} 1(v_s \sim v_t) - d$$
and by hypothesis, we know that $\sum_{s=1}^{t-1} 1(v_s \sim v_t) = d$,
so
$$\chi_t(v) = d + 1$$
and the relaxation process requires that $v$ topple again.

The second statement of the lemma is immediate.

\end{proof}
In particular, given a radical of the expanded cactus,
adding a grain to the root vertex causes each vertex to topple
at most once.

\begin{lemma}
\label{lem:connected-toppling-vertices}
Let $v \in G$ be given such that $\chi(v) = d$.
Let $v_1, \ldots, v_k$ be a toppling sequence for the
avalanche caused by adding a grain to $v$.  For each
$t = 1, 2, \ldots, k$, define $S_t = \{v_s|1 \leq s \leq t\}$.
The subgraph of $G$ induced by $S_t$ is connected
and contains $v$.
\end{lemma}
\begin{proof}
Let $t$, $1 \leq t \leq k$ be given.
Clearly, $v \in S$.  Let $H$ be the subgraph of $G$
induced by $S_t$, and suppose for a contradiction that
$H$ has more than one component.  Let $C$ be a component
of $H$ that does not contain $v$.
Choose the smallest positive integer $s$ such that $v_s \in S_t$.
By this choice of $s$, we have the no neighbor of $v_s$ has toppled
by the $s$th step of relaxation.  Therefore, $\chi_{s-1}(v_s) =
\chi(v_s) \leq d$, so we are not allowed to relax $v_s$ at the
$s$th step, which gives us a contradiction.  Hence,
$H$ is connected.
\end{proof}

The following is an immediate consequence.
\begin{lemma}
Let $G$ be a finite decorated rooted subtree with root vertex
$o$ having height $3$.  Then
the cells that topple in the avalanche caused by adding a grain
to $o$ form a cluster.
\end{lemma}

\begin{lemma}
Let $G$ be a finite decorated rooted subtree with root vertex
$o$ with configuration $\chi$ such that
$\chi(o) = 3$, and let $c$ be cell with origin-facing
vertex having height less than $3$ (in particular, $c$ is not
the origin cell).  Then $c$ does not topple
in the avalanche caused by adding a vertex to $o$.
\end{lemma}
\begin{proof}
Suppose that the avalanche caused by adding a vertex to $o$
has a toppling sequence $v_1, \ldots, v_k$.
Let $v$ be the origin-facing vertex of $c$.
We claim that $v$ does not topple.  Suppose for a contradiction
that $v_t = v$ for a (necessarily unique) positive integer $t$.
We note that $v$ is a cut-vertex of $G$, and only one neighbor
$u$ of $v$ in $G$ is contained in the connected component of
$G \setminus \{v\}$ containing $o$.  Thus, by Lemma
\ref{lem:connected-toppling-vertices}, we have that
$S_t = \{v_s|1 \leq s \leq t\}$ induces a connected subgraph of $G$.
Therefore, the only neighbor of $v$ that could topple before $v$
does is $u$.  But we then have
$\chi_{t-1}(v) = \chi(v) + 1 < d + 1$, so $v$ cannot be
the $t$th vertex to topple.  This is a contradiction.
So $v$ does not topple.
The other two vertices in $c$ are not in the same component
of $G \setminus \{v\}$ as $o$, and so they cannot topple either.
Hence, $c$ does not topple.
\end{proof}

The following transfer rule allows us to manipulate toppling
sequences in a way that makes the proof of the main theorem
easier.
\begin{lemma}
\label{lem:transfer}
Let $G$ with origin $o$ and configuration $\chi$ be given
such that $\chi(o) = d$.
Suppose that $v_1, \ldots, v_k$
is a toppling sequence where no vertex appears more than once.
If $1 \leq i < j \leq k$ and
for all $m$ such that $i \leq m < j$, we have $v_j \not\sim
v_m$, then the sequence
$v_1, \ldots, v_j, v_i, \ldots, v_{j-1}, v_{j+1}, \ldots, v_k$
is a toppling sequence.
\end{lemma}
\begin{proof}
We know that each vertex topples at most once, so each
vertex appears at most once in the toppling sequence.
In the original sequence, the number of grains on $v_j$
at the time it topples is equal to $\chi(v_j) +
\sum_{t = 1}^{j-1} 1(v_t \sim v_j)$.  Because none
of the vertices $v_i, \ldots, v_{j-1}$ are adjacent to $v_j$,
we have that the number of grains on $v_j$ before
$v_i$ topples is also $\chi(v_j) +
\sum_{t = 1}^{i-1} 1(v_t \sim v_j)$, which is the same
as the number of grains on $v_j$ at the time $v_j$ topples.
Consequently, moving $v_j$ before $v_i$ will still permit us
to legally topple $v_j$ at its new position in the sequence.
The other vertices $v_i, \ldots, v_{j-1}$ will have the same
number of grains when they topple in the new sequence as they
did in the old sequence.  Finally, the same vertices topple
in the same order, so after the topplings in each sequence,
the configurations are the same.  Therefore,
the new sequence gives a stable configuration after the topplings
are finished.
Thus, the new sequence,
$v_1, \ldots, v_j, v_i, \ldots, v_{j-1}, v_{j+1}, \ldots, v_k$,
is a toppling sequence.
\end{proof}

\begin{lemma}
\label{lem:recursion}
Let $G$ be a finite decorated rooted subtree with origin vertex
$o$ with configuration $\chi$ such that
$\chi(o) = 3$.  Let $c$ be a cell in $G$, and let
$v$ be the origin-facing vertex of $c$.  Let $e$ be the edge
connected $v$ to a vertex not in $c$, and define
$U$ to be the decorated rooted subtree that is
the connected component of $G \setminus \{e\}$ containing $v$.
If $v$ topples in
the avalanche caused by adding a grain to $o$, then,
for every vertex $u \in V(U)$, $u$ topples in the avalanche in $G$
caused by adding a grain to $o$ if and only if $u$ topples
in the avalanche in $U$ caused by adding a grain to $v$.
\end{lemma}
\begin{proof}
Let $\chi$ be a configuration on $G$, and let
$v_1, \ldots, v_k$ be a toppling sequence.
By hypothesis, $v$ is in this
sequence as vertex $v_\ell$.
First, we claim that there is a permutation
of this sequence such that all the vertices after $v$
in the sequence are in $V(U)$.  We use Lemma \ref{lem:transfer}
iteratively as follows: take the smallest index $i$ such
that $i > \ell$ and $v_i \notin V(U)$.  By Lemma
\ref{lem:connected-toppling-vertices}, we have that
the neighbor of $v$ that is not in $U$ appears before
$v$ in the toppling sequence.  Thus, since each vertex
topples at most once, we have $v_i \not\sim v$.  Moreover,
$v_i \not\sim w$ for any $w \in V(U)$, including
$v_{\ell+1}, \ldots, v_{i-1}$.  Thus, we can apply Lemma
\ref{lem:transfer} to bring $v_i$ directly before $v_\ell$.
After completing this process, we have a toppling sequence
of the form $w_1, \ldots, w_m, v, u_1, \ldots, u_{m'}$, where
each $w_i \notin V(U)$ and each $u_i \in V(U)$.  

Now, we want to show that $v, u_1, \ldots, u_{m'}$ is a toppling
sequence for the avalanche caused by adding a grain to the
configuration $\chi|_U$ at $v$ in $U$.  For each $i$, $1 \leq i \leq m'$,
we have that $u_i$'s neighbors are in $V(U)$.  Therefore,
$u_i$'s height just before it is toppled is the same in both
$w_1, \ldots, v, u_1, \ldots, u_{m'}$ and $v, u_1, \ldots, u_{m'}$.
Moreover, after toppling $v, u_1, \ldots, u_{m'}$, the resulting
configuration $\left(\chi|_U\right)_{m' + 1}$ must be stable,
for if this were not the case, we would be able to topple some vertex
$z \in U$, which we could have also toppled at the end of the
sequence $w_1, \ldots, w_m, v, u_1, \ldots, u_{m'}$.  Thus,
the sequence $v, u_1, \ldots, u_{m'}$ is a toppling sequence
for the avalanche caused by adding a grain to configuration
$\chi|_U$ at $v$ in $U$, and the lemma follows.
\end{proof}

\begin{lemma}
\label{lem:possible-forbidden-subconfigurations}
Let $\chi$ be a configuration on the decorated rooted subtree
$G$ with origin $o$.  Suppose $\chi(o) = 3$, and let
$v_1, \ldots, v_k$ be a toppling sequence for the avalanche
caused by adding a grain to $o$.  If $\chi$ has a forbidden
subconfiguration $S$, then no vertex in the sequence
$v_1, \ldots, v_k$ is in $S$.
\end{lemma}
\begin{proof}
To show this, it suffices to show that the vertices
$v_1, \ldots, v_k$ can be sequentially burned according to
the burning algorithm.  Indeed, $v_1$, which must be $o$,
can be burned because $\chi(o) = 3$ and $d(o) = 2$.
Now, for any positive integer $i$ from $2$ to $n$,
the number of unburnt neighbors of $v_i$ remaining
after we have burned $v_1, \ldots, v_{i-1}$ is
at most $3 - \sum_{j=1}^{i-1} 1(v_j \sim v_i)$.
But we know that $\chi(v_i) + \sum_{j=1}^{i-1} 1(v_j \sim v_i) > 3$.
This implies that $\chi(v_i) > 3 - \sum_{j=1}^{i-1} 1(v_j \sim v_i)$,
which is more than the number of unburnt neighbors of $v_i$.
Consequently, we can burn $v_i$ at that stage.  By the burning
algorithm, no forbidden subconfiguration $S$ can contain
any of the vertices $v_1, \ldots, v_k$
\end{proof}

\begin{proof}[Proof of Theorem \ref{thm:filling-rules}]
$(\Longrightarrow)$.  By strong induction on the number
of cells in $C$.  In the base case, suppose that $C$ is a one-cell
cluster.  Then the origin cell is a terminal cell.  Note that
every allowed cell configuration with origin-facing vertex having
height $3$ is allowed to be filled into a terminal cell.
We therefore check all five cases for the cell configuration
(up to permuting the two non-origin heights):
\begin{enumerate}
\item {\bf 3-3-3 and 3-3-2.}  The two non-origin vertices topple in the
avalanche, and so
the attached radicals must either be empty or have a root vertex
with height less than $3$.  In either case, the radicals
are stoppers, as prescribed by the filling rules.
\item {\bf 3-3-1.}  The non-origin vertex with height $3$ topples in
the avalanche, and so the attached radical must either be empty
or have a root vertex with height less than $3$, making such
radical a stopper.  Moreover, if the radical attached to the $1$
is not strongly allowed, then it forms a forbidden subconfiguration
when attached to the $1$.  Thus, the radical attached to the $1$
must be strongly allowed.  These are precisely the conditions
imposed by the filling rules.
\item {\bf 3-2-2.}  We need to show that at least one of the
attached radicals must be strongly allowed.  For if both are
weakly allowed, then the two radicals, when joined to the two
vertices with height $2$, have a forbidden subconfiguration.
\item {\bf 3-2-1.}  We must show that both radicals must be strongly
allowed.  If the radical attached to the $2$ is weakly allowed,
then that radical, union the two non-origin vertices of the origin cell,
contains a forbidden subconfiguration.  If the radical attached to the
$1$ is weakly allowed, then that radical, when attached to the $1$,
has a forbidden subconfiguration.
\end{enumerate}
Thus, the base case is complete.

Now for the induction step, we let $C$ be a cluster with $n \geq 2$ cells,
and we assume that $(\Longrightarrow)$ holds for any cluster
with fewer than $n$ cells.  Note that the origin cell of $C$
must be internal or medial.  We handle both cases in turn.

If the origin cell is medial, then exactly one non-origin vertex in
the origin cell is attached to another cell in $C$.  Call
that vertex $a$, and call the other non-origin vertex in the origin
cell $b$.  Denote by $c$ the vertex not in the origin cell
that is adjacent to $a$.
Let $U$ denote the decorated rooted subtree consisting
of $C$ minus the origin cell, rooted at vertex $c$.  Since $c$
topples, $a$ must topple.  The only four cell configurations for
the origin cell in which $a$ topples are 3-3-3, 3-3-2, 3-3-1, and
3-2-3 (where each triple of numbers denotes the heights of
$o$, $a$, and $b$, respectively).  We have four cases to check
to ensure that the radical attached to $b$ satisfies the filling rules:
\begin{enumerate}
\item {\bf 3-3-3, 3-3-2, and 3-2-3.}  Since $b$ topples,
the origin of the radical attached to $b$ must have height less
then $3$ (or else it would topple, too).  Thus, that radical
must be a stopper.
\item {\bf 3-3-1.}  If the radical attached to $b$ is weakly
allowed, then we have a forbidden subconfiguration when that
radical is attached to vertex $b$.  Thus, the radical attached
to $b$ must be strongly allowed.
\end{enumerate}
This means that the origin cell and the attached radical
must satisfy the filling rules.  Moreover, by Lemma \ref{lem:recursion},
when a grain
is added to the configuration $\chi|_U$ at $c$, the cells
that topple are precisely the cells in the cluster $U \cap C$
of $n-1$ cells.  Thus, by the induction hypothesis, those cells
and their attached radicals must satisfy the filling rules.
Therefore, the entire graph $G$ satisfies the filling rules with
respect to cluster $C$.

Now suppose that the origin cell is internal.  Then we must
have all three vertices in the origin cell topple.  This
means that the origin cell must be either 3-3-3 or 3-3-2 (with
the 2 in a non-origin vertex).  Removing the origin cell
induces two subradicals $U_1$ and $U_2$, and again by Lemma
\ref{lem:recursion}, the cells that topple in $U_1$ and $U_2$
are precisely $C \cap U_1$ and $C \cap U_2$, which are both
clusters of less than $n$ cells.  Those cells and their attached
radicals must satisfy the filling rules, and so the whole cluster
satisfies the filling rules.  This completes the induction
and the first half of the proof.

$(\Longleftarrow)$
For the second half of the proof, suppose that $G$ has a configuration
that satisfies the filling rules with respect to a cluster $C$ in $G$.
By inspection, it is clear that all of the cells in $C$ topple.
It remains to show that $G$ does not contain a forbidden subconfiguration.
Suppose for a contradiction that $S \subset V(G)$ is a forbidden
subconfiguration.  Then, we can take a connected component of $S$,
and that will also be a forbidden subconfiguration.  Note that $S$
cannot be a subset of a radical induced by $C$, as each radical must
be allowed.  Also, by Lemma \ref{lem:possible-forbidden-subconfigurations},
it follows that the vertices in $S$ are confined to the nontoppling vertices
in a single cell $c$ in the cluster and one or both of its attached radicals.
We examine each case of $c$:
\begin{enumerate}
\item {\bf $c$ is a 3-3-3, 3-3-2, or 3-2-3 medial cell.}  This is impossible,
as all of the vertices of $c$ topple.  Thus, $S$ is a subset of the attached
radical, which is impossible.
\item {\bf $c$ is a 3-3-1 medial cell.}  This implies that the attached
radical is weakly allowed, which is barred by the filling rules.
\item {\bf $c$ is a 3-3-3, 3-3-2, or 3-2-3 terminal cell.}  This is impossible,
as all of the vertices of $c$ topple.  Thus, $S$ is a subset of the attached
radical, which is impossible.
\item {\bf $c$ is a 3-3-1 or 3-1-3 terminal cell.}  Then $S$ must consist
of the nontoppling vertex with height $1$ plus some vertices in the attached
radical.  Thus, the attached radical is weakly allowed, when it is, in fact,
required to be strongly allowed by the filling rules.
\item {\bf $c$ is a 3-2-2 terminal cell.}  Then $S$ must contain both nontoppling
vertices with height $2$.  Let $U_1$ and $U_2$ be the attached radicals.
Then $S \cap V(U_1)$ and $S \cap V(U_2)$ have the property that every vertex
except the root has at least as many neighbors in $S$ as its height.
Therefore, $U_1$ and $U_2$ would each contain a forbidden subconfiguration
if its root were attached to a vertex with height $1$.  Thus,
$U_1$ and $U_2$ are both weakly allowed, which is not permitted by the filling
rules.
\item {\bf $c$ is a 3-2-1 or 3-1-2 terminal cell.}  We have two possibilities.
If the nontoppling vertex with height $2$ is not in $S$, then $S$ contains
the vertex with height $1$ and part of its attached radical, making that
attached radical weakly allowed.  But the filling rules do not permit this.
If the nontoppling vertex with height $2$ is in $S$, then $S$ must contain
the vertex with height $1$ also.  Then, the attached radical $U$ to
the vertex with height $2$ has a subset $V(U) \cap S$ satisfying
that every vertex in this subset, except for the root, has at least
as many neighbors in the subset as its height.  Therefore, if a vertex with
height $1$ is attached to $U$, this subset,
together with the newly-attached $1$, would form a forbidden subconfiguration.
Therefore, $U$ is weakly allowed.  Again, though, the filling rules do not
permit this.
\end{enumerate}
Therefore, we cannot find any such forbidden subconfiguration $S$.
Thus, as desired, $\chi$ is recurrent and adding a grain to the origin
causes the cells in $C$ to topple and no others.  This completes the proof.
\end{proof}

\subsection{Overall impact of filling rules}
The requirements for each cell configuration are
independent, and the requirements for each radical depend only on the cell to which the underlying
decorated rooted subtree is attached.  With the exception of a 3-2-2 terminal filling,
each assignment of a configuration to each cell in $C$
accounts for $\prod_{i=1}^{n+2} N_s(U_i)$ configurations,
since each radical $U_i$ must be strongly allowed
or a stopper, either of which is a restriction
that allows for $N_s(U_i)$ configurations on the radical.

However, a 3-2-2 terminal filling has a different effect on the number of allowed radicals on
the two attached decorated rooted subtrees.  Since neither radical has to be a stopper, but only one
has to be strongly allowed, this allows for three possibilities: both radicals are strongly allowed,
the first is weakly allowed and the second strongly allowed, or vice versa.  Since weakly allowed
radicals are twice as prevalent as strongly allowed radicals in the limit, this means that there
are 5 times as many configurations for the two attached radicals as there would be for any other
terminal filling.

Therefore, we can state that if we have a configuration of the cells of $C$
where there are $k$ terminal cells in $C$ with the 3-2-2 configuration, then
the number of overall configurations accounted for by the given configuration of $C$ is
$$5^k \prod_{i=1}^{n+2} N_s(U_i)$$

An easy and appropriate way to account for this fact is to treat the possibility of a 3-2-2
terminal filling as if it were five separate fillings.  If we count a 3-2-2 terminal fillings
as 5 \emph{effective fillings}, and all other fillings (internal, medial, and terminal)
as 1 effective fillings, then we have
\begin{enumerate}
\item 3 effective fillings for each internal cell
\item 4 effective fillings for each medial cell
\item 12 effective fillings for each terminal cell
\end{enumerate}
and
the total number of recurrent configurations in which
topplings occur in each cell in $C$ and no other cells is
$$3^{i} 4^{m} 12^{t} \prod_{k=1}^{n+2} N_s(U_k)$$
where $i, m, t$ are the number of internal, medial, and terminal cells of $C$.

\section{Extending the filling rules to finite copies of the expanded
cactus, generally}

In the previous section, we showed that the filling rules enumerate
all recurrent configurations on a decorated rooted subtree in which
adding a grain to the origin causes exactly $n$ cells to topple.
However, our main interest is to consider all recurrent configurations
on a sufficiently large finite copy of the expanded cactus
where adding a grain to the origin causes exactly $n$ cells to topple.
Although this extension may seem like it can be easily solved, the
extension to the expanded cactus presents complications that make
enumerating the configurations much more difficult.

\subsection{The liberty rule}
Let $G$ be a finite expanded cactus with origin $o$.  If $C$ is a
cluster about the origin that does not contain the cell directly
opposite $o$, then the height of the neighbor of $o$ not in the
origin cell is less than $3$.  Thus, no cell topples more than once,
and the filling rules can be used\footnote{Initially, we make a
small adaptation in that the radical attached to $o$ can be filled
by any stopper.}
to determine which stable configurations
have the property that adding a grain to $o$ causes precisely the cells
in $C$ to topple.

The issue is that positions resulting from applying the filling rules
no longer have to be recurrent.  To fix this, we impose an additional
requirement on our fillings.  In a configuration $\chi$ define a
\emph{liberty} to be a radical induced by $C$ that is attached to a
vertex in $C$ that is connected to $o$ by a path of vertices
of height $3$.  Note that the number of liberties depends on
the configuration.  We then require that at least one liberty be
strongly allowed.\footnote{Note that the filling rules already require
that a liberty be a stopper.}
This is justified by the following lemma.

\begin{lemma}
\label{lem:liberty-rule-1}
Let $G$ be a finite expanded cactus with origin $o$, and let
$C$ be a cluster of $G$ about the origin that does not contain
the cell directly opposite $o$.  If $\chi$ is obtained by applying
the filling rules to $C$ and its attached radicals (where the radical
attached to $o$ is a stopper), then $\chi$ is a recurrent configuration
if and only if it has at least one liberty that is a strongly allowed.
\end{lemma}
\begin{proof}
$(\Longrightarrow)$  Suppose $\chi$ is recurrent.  Let $D$ be the
subcluster of $C$ consisting of all cells in $C$ whose origin-facing
vertex is connected to $o$ by a path of vertices with height $3$.
Also, note that every
liberty of $\chi$ must be attached to a vertex in $D$.
Now take $S$ to be the set of all vertices in $D$ and the liberties
of $\chi$ (with respect to $C$).  If every liberty of $\chi$
were weakly allowed, then $S$ would form a forbidden subconfiguration.
Thus, there must at least one liberty that is strongly allowed.

$(\Longleftarrow)$  Suppose that $U$ is a strongly allowed liberty of
$\chi$.  Apply the burning algorithm.  Since $U$ is strongly allowed,
every vertex in $U$ must burn.  Now, since $U$ is a liberty,
the vertex $v$ to which $U$ is attached burns, and there exists
a $vo$-path where each vertex has height $3$---all of those vertices
burn sequentially.  Once the origin burns, its attached radical must
burn (since it is allowed).  The remaining vertices in $G$ burn
according to Lemma \ref{lem:possible-forbidden-subconfigurations}
and the proof of Theorem \ref{thm:filling-rules}.  Therefore, $\chi$ is
recurrent.
\end{proof}

\subsection{The breakthrough problem}
The filling rules work when each vertex in the graph can only topple once,
such as in the decorated rooted subtree.  However, when the origin $o$
has degree $d$ in $G$, it can topple more than once.  We can arrange
the toppling sequence for the avalanche caused by adding a grain to $o$
so that we only topple $o$ if it is the only vertex with more than $d$
grains.  At any time we are required to topple $o$, it must have $d+1$ grains,
and so the unstable configuration $\psi$ at this time is formed
from a stable configuration $\psi'$ by adding a grain to $o$.
Then, by Lemma \ref{lem:origin-first-repeated-toppling-vertex}
and using the arrangement described above, we can create a toppling sequence
of the form
$$o, v_{11}, \ldots, v_{1n_1}, o, v_{21}, \ldots, o, v_{k1}, \ldots, v_{kn_k}$$
where for each $i$ from $1$ to $k$, no vertex is repeated in the sequence
$v_{i1}, \ldots, v{in_i}$.

In this way, we can envision an avalanche as a series of ``waves'' of
vertices that topple between successive topplings of $o$.  In the case
of a tree, no new vertices topple after the first wave.  This is
a consequence of the following lemma:

\begin{lemma}
\label{lem:tree-first-wave}
Let $G$ be a tree with root $o$ that is a finite subgraph of infinite connected
$d$-regular graph $\Gamma$.  Let $\chi$ be a configuration on $G$
such that $\chi(o) = 3$.  Suppose that $v \neq o$ topples in the avalanche
caused by adding a grain to $\chi$ at $o$.  Let
$$o, v_{11}, \ldots, v_{1n_1}, \ldots, o, v_{21}, \ldots, o, v_{k1}, \ldots, v_{kn_k}$$
be the toppling sequence for the avalanche caused by adding a grain to $\chi$ at $o$
organized into waves as described above (i.e., $o$ topples if and only if it
is the only vertex with height greater than $d$).  Then
$v \in \{v_{11}, \ldots, v_{1n_1}\}$.
\end{lemma}
\begin{proof}
By induction on the level of $v$.  If $v$ is level $1$, then, by Lemma
\ref{lem:origin-multiple-topplings}, $o$ can only topple a second time
after $v$ has already toppled; thus, the lemma holds.
Now suppose the lemma holds for all vertices $v$ whose level is less than $n$,
and let $v$ be a vertex of level $n$.  If $\chi(v) = d$ and $v$ topples,
then it topples as soon as its neighbor $u$ on the unique $vo$-path topples.
By the induction hypothesis, $u$ topples in the first wave.  Therefore, after
$u$ topples, we have that $v$ has $d+1$ grains and can topple, meaning
that it is among the vertices toppling in the first wave.
So the inductive step holds if $\chi(v) = d$.  However, if $\chi(v) < d$,
then, since $v$ ultimately topples (meaning that it eventually gets $d+1$ grains),
its neighbor $u$ on the unique $vo$-path must topple twice (see Lemma
\ref{lem:connected-toppling-vertices}).  By the induction hypothesis,
$u$ topples in the first wave.
Now, consider the configuration $\psi$ reached after toppling the vertices
in the new toppling sequence up to (but not including) the second toppling of $o$.
Note that $\psi(u) < d$ ($u$ gained at most $d-1$ grains from toppling of
neighbors other than $v$ and lost $d$ grains when it toppled) and
$\psi$ is the unstable configuration reached by adding a grain to
stable configuration $\psi'$ at $o$, where $\psi'$ is $\psi$ minus a grain at $o$.
Note that
$$o, v_{21}, \ldots, o, v_{k1}, \ldots, v_{kn_k}$$
is a toppling sequence for the avalanche caused by adding a grain to $\psi$ at
$o$.  By the induction hypothesis, $u$ is among the vertices
$v_{21}, \ldots, v_{2n_2}$ that topple in the first wave of this new
sequence.  Also, we apply Lemma \ref{lem:connected-toppling-vertices}
to show that exactly one neighbor of $u$ topples prior to $u$ in
the first wave of this new toppling sequence.  Thus, when $u$ is scheduled
to topple, it has at most $d-1 + 1 = d$ grains.  But this is impossible.
Therefore, $u$ cannot topple twice.  Thus, $v$'s toppling implies
that $\chi(v) = d$, which is the case we have already handled,
and so the induction is complete.
\end{proof}

Thus, on trees, any vertex either topples in the first wave of an avalanche or
never.  This is not true, in general, for cacti.  To see this, consider
the following configuration on a subgraph of the expanded cactus.

\begin{figure}[htbp]
\begin{center}
\includegraphics[height=0.3\textheight]{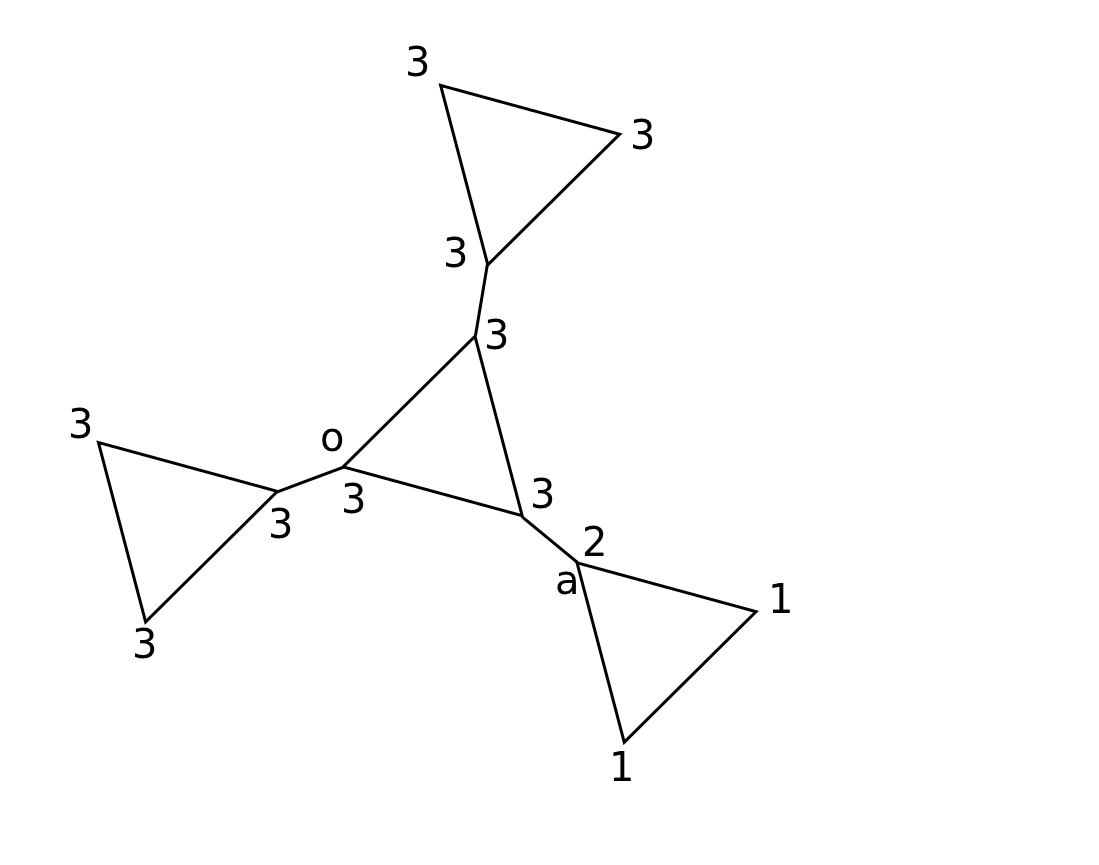}
\end{center}
\caption{A example of a configuration where a vertex (here, $a$)
does not topple in the first wave but does ultimately topple.}
\label{fig:sandpile-cactus-counterexample}
\end{figure}

Therefore, the filling rules described in Section \ref{sec:enumerating-the-configurations}
need modification to correctly enumerate the recurrent configurations where the cells
in a specified cluster topple.  One possible approach is to divide a finite expanded
cactus $G$ with origin $o$ into two rooted subtrees $U_1$ and $U_2$
by cutting the edge connecting $o$ to its neighbor $o'$ not in the origin cell.
Then, a configuration $\chi$ on $G$ can be associated with an ordered pair of
configurations $(\chi_1, \chi_2)$ on $U_1$ and $U_2$, respectively,
and the process of adding a grain to $o$ and relaxing the avalanche that
occurs (if there is one)
can be performed in the following manner:
\begin{enumerate}
\item Let $i = 1$, $\chi_1^0 = \chi_1$, and $\chi_2^0 = \chi_2$.
\item Add a grain to $\chi_1^{(i-1)}$ at $o$ to produce
configuration $\chi_1^{(i-1)}{}'$.  If $\chi_1^{(i-1)}{}'(o) \leq 3$, 
set $\chi_1^{(i)} = \chi_1^{(i-1)}{}'$ and
$\chi_2^{(i)} = \chi_2{(i-1)}$, and go to Step 8.
\item Relax $\chi_1^{(i-1)}{}'$ in $G_1$ to a stable configuration $\chi_1^{(i)}$.
\item Add a grain to $\chi_2^{(i-1)}$ at $o'$ to produce
configuration $\chi_2^{(i-1)}{}'$.  If $\chi_2^{(i-1)}{}'(o') \leq 3$, 
set $\chi_2^{(i)} = \chi_2^{(i-1)}{}'$, and go to Step 8.
\item Relax $\chi_2^{(i-1)}{}'$ in $G_1$ to a stable configuration $\chi_2^{(i)}$.
\item Increment $i$.
\item Go to step 2.
\item Stop; the avalanche is finished.  The resulting stable configuration
is associated with the ordered pair $(\chi_1^{(i)}, \chi_2^{(i)})$.
\end{enumerate}

\subsection{First-wave topplings and the filling rules}
From the procedure above, it is possible to determine which vertices
topple in the first wave of an avalanche caused by adding a grain to $o$.
Let $\chi$ be a configuration on finite expanded cactus $G$ with origin $o$
with $(\chi_1, \chi_2)$ the ordered pair of configurations on 
decorated rooted subtrees $U_1$, $U_2$,
as defined above.  Then the vertices that topple in an avalanche (if any)
caused by adding a grain to $\chi$ at $o$ are as follows:
\begin{enumerate}
\item No vertices, if $\chi_1(o) < 3$.
\item The vertices that would topple in $U_1$ in the avalanche caused by
adding a grain to $\chi_1$ at $o$,
if $\chi_1(o) = 3$ and $\chi_2(o') < 3$.
\item The vertices that would topple on $\chi_1$ in the avalanche caused by
adding a grain to $\chi_1$ at $o$, plus the vertices that would topple
on $\chi_2$ in the avalanche caused by adding a grain to $\chi_2$ at $o'$,
if $\chi_1(o) = \chi_2(o') = 3$.
\end{enumerate}

Therefore, to enumerate the recurrent configurations in which the cells in a
given cluster $C$ about the origin topple in the first wave
when a grain is added at $o$, we must do the following:
\begin{enumerate}
\item If $C$ does not contain the cell attached directly to $o$, fill
$U_1$ according to the filling rules, have $U_2$ be a stopper, and ensure
that at least one liberty is a strongly allowed stopper.
\item If $C$ contains the cell attached directly to $o$, fill
$U_1$ and $U_2$ according to the filling rules, and ensure that
at least one liberty is a strongly allowed stopper.
\end{enumerate}

If $C$ is a cluster containing the origin, we say that $\chi$
satisfies the \emph{extended filling rules} with respect to $C$
if
\begin{enumerate}
\item $C$ does not contain the cell directly attached to $o$,
$\chi_1$ satisfies the filling rules with respect to $C$,
and $\chi_2$ is a stopper, \emph{or}
\item $C$ contains the cell directly attached to $o$,
$\chi_1$ satisfies the filling rules with respect to $C \cap U_1$,
and $\chi_2$ satisfies the filling rules with respect to $C \cap U_2$.
\end{enumerate}
We say that such a configuration $\chi$ \emph{satisfies the liberty rule}
with respect to $C$ if
at least one liberty in $\chi$ is strongly allowed.

We want to prove the following result:
\begin{thm}
\label{thm:filling-rules-2}
Let $\chi$ be a configuration on $G$, a finite expanded cactus,
and let $C$ be a cluster of $G$.
Then $\chi$ is a recurrent configuration and
adding a grain to the origin vertex causes exactly the cells in $C$ to
topple in the first wave
if and only if $\chi$ satisfies the extended filling rules and
the liberty rule with respect to $C$.
\end{thm}

We first extend Lemma \ref{lem:liberty-rule-1} to all clusters:
\begin{lemma}
\label{lem:liberty-rule-2}
Let $G$ be a finite expanded cactus with origin $o$, and let
$C$ be a cluster of $G$ about the origin.
If $\chi$ satisfies the extended filling rules with respect to $C$
and its attached radicals, then $\chi$ is a recurrent configuration
if and only if it has at least one liberty that is a strongly allowed.
\end{lemma}
The proof parallels that of Lemma \ref{lem:liberty-rule-1}.

Theorem \ref{thm:filling-rules-2} follows from Theorem \ref{thm:filling-rules},
Lemma \ref{lem:liberty-rule-2}, and our characterization of the vertices
that topple in the first wave of an avalanche on a finite expanded cactus
caused by adding a grain at $o$.

\section{Counting the configurations}
Let $(G_i)_{i=1}^\infty$ be an increasing graph sequence exhausting
infinite expanded cactus $\Gamma$, and let $\epsilon > 0$ be given.
By Theorem 1, for each positive integer $n$,
we can choose a graph $G_N$ in the sequence
such that for every cluster $C$ about the origin with $n$ cells,
every induced decorated rooted subtree $U$ has $2 > x_U > 2 - \epsilon$.
To examine the limiting behavior of the avalanche dynamics, we consider
a sufficiently large expanded cactus to not only contain every cluster of size
up to $N$ but also ensure that every induced decorated rooted subtree
$U$ satisfies that $x_U$ arbitrarily close to $2$.

Ignoring the liberty rule,
the problem of enumerating the total number
of recurrent configurations on a sufficiently large copy of the expanded cactus
in which exactly $n$ cells topple in the first wave, divided by $\prod_{i=1}^{n+1} N_s(U_i)$,
is reduced to finding the total weight of animals of $n$
cells on a rooted 3-regular tree where an animals weight is the product of 3 for each internal vertex,
4 for each medial vertex, and 12 for each terminal vertex (corresponding to the number
of effective fillings).

The sum of the weights for rooted
animals with $n$ cells is given by the $x^n$ term in the power series expansion of the generating
functions $f(x)$ satisfying
$$f(x) = 12x + 8xf(x) + 3x[f(x)]^2$$
Also, we can write an equation determining $g(x)$, the generating function whose $x^n$ term gives
the sum of the weights of all animals with $n$ vertices:
$$g(x) = f(x) + [f(x)]^2$$

To help determine asymptotic behavior of the coefficients of $x^n$ for $f(x)$, and $g(x)$,
we use properties of the generating functions, especially their singularities (a good source
for the methods involved is \cite{Odlyzko95}).

First, using the quadratic formula, we can solve directly for $f(x)$, as follows:
$$3x[f(x)]^2 + (8x - 1)f(x) + 12x = 0$$
$$f(x) = \frac{1 - 8x - \sqrt{1 - 16x - 80x^2}}{6x}$$
Note that we choose the negative square root to ensure that $f(x)$ is defined at $x = 0$.

For $x > 0$, the smallest $x$ for which $f$ is not analytic is $\frac{1}{20}$; since all the coefficients
of $f$ are nonnegative real numbers, $f$ is analytic for $|x| \leq \frac{1}{20}$ except for $\frac{1}{20}$.
\cite{Odlyzko95}

Because of this, we can use the lemma of Polya:
\begin{lemma}[\cite{PolyaRead87} (see also \cite{HararyRobinsonSchwenk75})]
Let $f(x)$ be a generating function with the power series expansion
$$f(x) = a_0 + a_1 x + a_2 x^2 + \cdots$$
with radius of convergence $\alpha$.  Further, suppose that $f$ is analytic for each $x$ with $|x| \leq
\alpha$ except for $x = \alpha$.  Suppose that there exist functions $g$, $k$ analytic and regular in
a neighborhood of $\alpha$, with $g(\alpha) \neq 0$, such that
\begin{equation}
\label{eqn:polya-lemma-reqs}
f(x) = \left(1 - \frac{x}{\alpha}\right)^{-s} g(x) + \left(1 - \frac{x}{\alpha}\right)^{-t} k(x)
\end{equation}
where $s$ is not a nonpositive integer, and $t = 0$ or $t < s$.  Then the coefficients
$\{a_n\}_{n=1}^\infty$ satisfy the asymptotic relation
$$a_n \sim \alpha^{-n} n^{s-1} \frac{g(\alpha)}{\Gamma(s)}$$
\end{lemma}

The functions $f$ and $g$
can be written in the form of Equation \ref{eqn:polya-lemma-reqs}, allowing for asymptotic expansion.  In
particular
$$f(x) = -\frac{\sqrt{1 + 4x}}{6x} (1 - 20x)^{-\left(-\frac{1}{2}\right)} + \left(\frac{1}{6x} - \frac{4}{3}
\right)(1 - 20x)^0$$

Also, $$g(x) = \frac{5x - 1}{18x^2} \sqrt{1+4x} (1 - 20x)^{-\left(-\frac{1}{2}\right)}
-\left(\frac{16}{9} - \frac{13}{18x} + \frac{1}{18x^2}\right)$$

Therefore, the coefficients of the power series expansions of
both $f(x)$ and $g(x)$
are asymptotic
$K 20^n n^{-\frac{3}{2}}$ for an appropriate constant $K$
in each case.  Therefore, not taking into
account the liberty rule,
the fraction of recurrent configurations
in which topplings occur in exactly $n$ cells
when a grain is added to the origin
is asymptotic to
$$\frac{K 20^n n^{-\frac{3}{2}}}{5 \cdot 20^n} = K' n^{-\frac{3}{2}}$$

Moreover, we can create a sufficiently large expanded cactus,
and the results do not change based on which graph sequence was used.
This means that the cell-wise first-wave
critical exponent does not depend on the graph sequence we choose.

\section{Accounting for the liberty rule}

Ignoring the liberty rule results in overcounting the number of recurrent configurations in which exactly $n$ cells
topple by including non-recurrent configurations in the total.  Fortunately, it is possible to account for
the liberty rule as well.

Each configuration $\chi$ has a given number of liberties.
Recall that the substitution rules specify that among the liberties,
at least one
must be a strongly allowed stopper.  Since $\frac{7}{20}$ of the strongly allowed radicals are stoppers
as well (see Section \ref{sec:enumerating-the-configurations}),
the fraction of the configurations of the whole graph fitting $\chi$
compliant with the liberty rule is
$$\left(1 - \frac{13}{20}\right)^{\ell(\chi)}$$
where $\ell(\chi)$ is the number of liberties of the cluster configuration $\chi$.

For a cluster $C$, define $$\phi(C) = \frac{\sum_{\chi} 5^{s(\chi)} \left[1 - \left(
\frac{13}{20}\right)^
{\ell(\chi)}\right]}{\sum_{\chi} 5^{s(\chi)}}$$ where $\chi$ ranges over all configurations of $C$
that satisfy the extended filling rules, and $s$ maps each function $\chi$ to the
number of terminal cells that have a 3-2-2 filling.
For each $n$, define $\phi_n$ by
$$\phi_n = \frac{\sum_{C} \phi(C) 3^{i(C)}4^{m(C)}12^{t(C)}}{\sum_{C} 3^{i(C)}4^{m(C)}12^{t(C)}}$$
where $C$ runs over all clusters of size $n$, and $i(C), m(C), t(C)$ denote the number of internal,
medial, and terminal cells of $C$.

Recall that $g(x) = \sum_{n=0}^\infty b_n x^n$, where $b_n$ is the weighted number of clusters about
the origin of size $n$.  Define $h(x) = \sum_{n=0}^\infty \phi_n b_n x^n$, where $\phi_n b_n$ is the
number of recurrent configurations
in which adding a grain to the origin causes exactly $n$ cells to topple in the first wave.

We claim that $\phi_n > \frac{7}{48}$ for each $n$.  To do this, we show that at most
$\frac{7}{12}$ of the configurations of clusters satisfying the extended filling rules have zero liberties.
Let $\chi$ be a configuration of cluster $C$ with zero liberties satisfying the extended filling rules.  We will use the following method
to arrive at a terminal cell connected by a path of vertices with height 3 to the origin.  We start at
the origin cell.  From there, we will construct a connected chain of cells, all of which have a vertex
that can be connected to the origin by a path of vertices with height 3.  Given a chain of cells,
if the last cell is not terminal (satisfying our goal), we choose the next cell as follows:
if the last cell is medial, then we choose the cell in $C$ farther from the origin; that cell is connected
to the origin by a path of vertices with height 3 because the only valid fillings for the previous
cell possible in a zero-liberty configuration $\chi$ are 3-3-1 and 3-3-2, with the 3's facing the
two adjoining cells in $C$.  If the last cell is internal, then it is either 3-3-3 or 3-3-2, in which
case we arbitrarily select a cell not already in our chain directly attached to a 3 from the previous cell.
Because the expanded cactus is treelike and $C$ is finite, we must arrive at a terminal cell.
This terminal cell must either be 3-2-1 or 3-2-2.  However, we can replace the selected terminal cell
with 3-3-1, 3-3-2, or 3-3-3, giving the new configuration at least one liberty.  For every seven
zero-liberty configurations on $C$ satisfying the extended filling rules, there are five configurations on $C$ with at least one liberty also satisfying the extended filling rules.
Therefore, at most $\frac{7}{12}$ of the configurations on $C$ have zero liberties.  Thus,
$$\phi(C) \geq \frac{5}{12}\frac{7}{20} = \frac{7}{48}$$
Since this holds for all clusters $C$ of any size, we have $\phi_n \geq \frac{7}{48}$.  Also,
it follows directly from the definition of $\phi$ that $\phi_n \leq 1$ for each $n$.  Thus,
$\phi_n$ is bounded between two positive numbers.

Therefore, the number of recurrent configurations
in which exactly $n$ cells topple in the first wave is bounded between
$\frac{7}{48}$ and $1$ times a function asymptotic to $Kn^{-\frac{3}{2}}$.  This means that the
cell-wise first-wave critical
exponent of the expanded cactus is $\frac{3}{2}$,
and the proof
of this paper's theorem is established.

\section{Open questions and extension to arbitrary transitive decorations}
The first open question is to examine what the asymptotic behavior is
on the expanded cactus
for the total masses of avalanches that include more than one wave.
In decorating a tree, we lose the important property of trees that
any vertices that topple in an avalanche must topple in the first wave.
Evaluating how the possibility that new vertices may topple in a second or
subsequent wave
affects the asymptotic behavior of avalanche mass in a nontrivial fashion.
Thus, for the purposes of gathering evidence for mean-field theory,
it might be better to consider measures like first-wave critical exponent
that are easier to determine than the critical exponent for graphs like the
cactus.

In Section 1, we noted the following mean-field conjecture
\begin{conj}
Let $\Gamma$ be an infinite $d$-regular tree
and $\Gamma'$ be a graph quasi-isometric
to $\Gamma$ (including, but not limited to, a decoration of $\Gamma$).
Then $\Gamma'$ has the same critical exponent as $\Gamma$, namely $\frac{3}{2}$.
\end{conj}

We can propose a similar mean-field conjecture for cell-wise first-wave critical exponents:
\begin{conj}
\label{con:quasi-isometric-first-wave-critical}
Let $\Gamma$ be an infinite $d$-regular tree
and $\Gamma'$ be a graph quasi-isometric
to $\Gamma$ (including, but not limited to, a decoration of $\Gamma$).
Then $\Gamma'$ has the same cell-wise first-wave critical exponent as 
the first-wave critical exponent of $\Gamma$, namely $\frac{3}{2}$.
\end{conj}
(Recall that, on trees, the first-wave critical exponent is equal to the critical exponent
by Lemma \ref{lem:tree-first-wave}.)

Let $\Gamma$ be an infinite $d$-regular tree, and let $F$ be a connected
transitive graph on $d$ vertices.
(This means that $F$ is $k$-regular for some positive integer $k$.
We only consider such graphs $F$ for which $k > 1$;
in particular, we prohibit $F$ from being $K_2$.)
Then, if we define $\Gamma'$ to be the decoration of $\Gamma$ where each
vertex is replaced by a copy of $F$, we have that $\Gamma'$ is transitive
and $k+1$-regular.  We can consider the first-wave avalanche cell mass on $\Gamma'$.
We can then enumerate
all possible allowed cell configurations
and come up with tables similar to Tables \ref{tab:recurrent-configuration-origin}
and \ref{tab:recurrent-configuration-radical} in
Section \ref{sec:enumerating-recurrent-configurations} that indicate the following:
\begin{itemize}
\item for a given origin cell configuration, what combinations of weakly-allowed
and strongly-allowed radicals result in a recurrent configuration
\item for a given cell configuration
that is at the root of a decorated rooted subtree,
what combinations of weakly-allowed and strongly-allowed radicals
result in weakly-allowed and strongly-allowed whole radicals
\end{itemize}
We can compute the limiting weak-to-strong ratio $x$.
Using the same methods as in the expanded cactus,
we can show that for any positive integer $n$
and for any cluster $C$ about the origin of $n$ cells,
there is a deep enough finite subgraph $G$ of
$\Gamma$ that are the union of cells, where the ratio of
the number of recurrent configurations of $G$ to the product of the
number of strongly-allowed radicals induced by $C$ grows as $D^n$
for some $D \in \mathbb{Z}^+$.  This gives the exponential factor in
the denominator.

To calculate the exponential factor in the numerator, we use another counting
method analogous to the extended filling rules for the expanded cactus.
The same filling rule principles in the expanded cactus apply to the
general case:  each cell in the cluster must be filled in with
an allowed cell configuration
corresponding to its position within the cluster, and apart from the liberty
rule (which still holds), the requirements of the cluster's induced radicals
are completely determined by the configuration of the cell to which they
are directly attached.  In the general case, instead of internal, medial,
and terminal cells, there are $2^{d-1}$ classes of cells
in the cluster, identified by subsets of the set $\{2, 3, \ldots, d\}$.
For each cell in the cluster, reindex the vertices of $F$
so that the closest vertex in the cell to the origin vertex of $G$ has
index $1$.  Then the cell belongs to the class
$$\{i \in \{2, 3, \ldots, d\}| \textrm{a radical is attached to vertex }
 i \textrm{ of the cell}\}$$

We then use logic to calculate for each $S \subset \{2, 3, \ldots, d\}$,
the allowed cell configurations for a cell in class $S$ and the restrictions
on the attached radicals, if any.  Using these rules, we can compute
a weighting factor $a_S$ that gives the number of effective fillings.
To count the coefficient in the numerator, we use generating functions
to count the number of weighted clusters.  The generating functions $f(x)$
and $g(x)$ in the expanded cactus case can be adapted for this purpose.
We have
$$f(x) = x\left[\sum_{S \subset \{2, 3, \ldots, d\}} a_S (f(x))^{d-1-|S|}
\right]$$
The same argument in the expanded cactus that $g(x)$, the generating
function for the number of weighted clusters containing the origin, tells us
that the formula $g(x) = f(x) + (f(x))^2$ applies unchanged to the general case.

If Conjecture \ref{con:quasi-isometric-first-wave-critical} is true,
the following conjecture should hold regarding $g(x)$.
\begin{conj}
\label{con:numerator-denominator}
The generating function $g(x)$ has its closest singularity at the origin
at $x = \frac{1}{D}$.  Moreover, near $\frac{1}{D}$, the dominant term in
$g(x)$ is $(1 - Dx)^{\frac{1}{2}}$.
\end{conj}
This would immediately imply that the critical exponent of $\Gamma'$ is
$\frac{3}{2}$.

If Conjecture \ref{con:numerator-denominator} holds,
then the two separate methods of counting configurations
give an exponential factor of $D$ in the denominator and an generating function
with a singularity at $\frac{1}{D}$.
This brings up the following open question:
\begin{open}
Is there a combinatorial argument for why Conjecture 2 holds in the general case
of decoration of a $d$-regular tree by a connected transitive graph with $d$
vertices?
\end{open}

The two exponential factors in the numerator and the denominator are
computed using different methods.  However, these methods currently
rely on examining all recurrent cell configurations.  The
number of cell configurations is $d^{k+1}$, where $d$ is the number
of vertices in $F$, the decorating $k$-regular graph.  Of course,
some of these might be forbidden, but exhaustive examination of
each individual recurrent cell configuration quickly becomes prohibitive
for larger graphs.  Therefore, the best decorating graphs to examine next
are graphs whose symmetry we can leverage to reduce the number of cases
to examine, such as complete graphs.

In the case of a complete graph $K_d$, one can see that the allowed
cell configurations are precisely those satisfying that
there does not exist a positive integer $m$ such that more than $m$
vertices in the cell has height less than or equal to $m$.
Moreover, when computing the exponential factor in the numerator,
we can use automorphisms of the vertices in a cell and the
attached radicals to conclude that, for every $m$ with $0 \leq m \leq d-1$,
the weighting factors $a_S$ are equal for every
$S \subset \{2, 3, \ldots, d\}$ with $|S|=m$.  Therefore, when
formulating an analogue of the filling rules for these cactus
graphs, we can classify cells in a cluster based on the number
of attached radicals without having to reference their positions,
just as in the case of the expanded cactus (decorated by $K_3$).

That said, if Open Question 1 is answered in the affirmative, then
we can conclude that the broad conjecture that every cactus graph
has cell-wise first-wave critical exponent $\frac{3}{2}$ is true without having
to engage in prohibitive complex case analysis.

\bibliographystyle{amsalpha}
\bibliography{sandpilebib}

\providecommand{\bysame}{\leavevmode\hbox to3em{\hrulefill}\thinspace}
\providecommand{\MR}{\relax\ifhmode\unskip\space\fi MR }
\providecommand{\MRhref}[2]{%
  \href{http://www.ams.org/mathscinet-getitem?mr=#1}{#2}
}
\providecommand{\href}[2]{#2}
\begin{thebibliography}{MRZ01}

\bibitem[Dha99]{Dhar99}
Deepak Dhar, \emph{Studying self-organized criticality with exactly solved
  models}, Pre-print from arXiv, available at {\tt
  http://arXiv.org/abs/cond-mat/9909009v1}.

\bibitem[DM90]{DharMajumdar90}
Deepak Dhar and S.~N. Majumdar, \emph{Abelian sandpile model on the {B}ethe
  lattice}, J Phys. A \textbf{23} (1990), 4333--4350.

\bibitem[FE61]{FisherEssam61}
Michael~E. Fisher and John~W. Essam, \emph{Some cluster size and percolation
  problems}, J. Math. Phys. \textbf{2} (1961), no.~4, 609--619.

\bibitem[HRS75]{HararyRobinsonSchwenk75}
Frank Harary, Robert~W. Robinson, and Allen~J. Schwenk, \emph{Twenty-step
  algorithm for determining the asymptotic number of trees of various species},
  J. Austral. Math. Soc. (Series A) \textbf{20} (1975), 483--503.

\bibitem[Koz08]{Kozakova08}
Iva Koz{\'{a}}kov{\'{a}}, \emph{Critical percolation of free product of
  groups}, Internat. J. of Algebra and Comput. \textbf{18} (2008), no.~4,
  683--704.

\bibitem[MN11]{MatterNagnibeda11}
M.~Matter and T.~Nagnibeda, \emph{Abelian sandpile model on randomly rooted
  graphs and self-similar groups}, arXiv, May 2011, Pre-print from arXiv,
  available at {\tt http://arxiv.org/abs/1105.4036}.

\bibitem[MRZ01]{MeesterRedigZnamenski01}
Ronald Meester, Frank Redig, and Dmitri Znamenski, \emph{The abelian sandpile;
  a mathematical introduction}, Markov Processes and Related Fields \textbf{7}
  (2001), 509--523.

\bibitem[Odl95]{Odlyzko95}
A.~M. Odlyzko, \emph{Asymptotic enumeration methods}, Handbook of Combinatorics
  (R.~L. Graham, M.~Gr{\"{o}}tschel, and L.~Lov{\'{a}}sz, eds.), vol.~2, The
  MIT Press, 1995, pp.~1063--1229.

\bibitem[PR87]{PolyaRead87}
G.~P{\'{o}}lya and R.~C. Read, \emph{Combinatorial enumeration of groups,
  graphs, and chemical compounds}, Springer-Verlag, 1987.

\bibitem[ZLS95]{ZapperiLauritsenStanley95}
Stefano Zapperi, Kent~B{\ae}kgaard Lauritsen, and H.~Eugene Stanley,
  \emph{Self-organized branching processes: Mean-field theory for avalanches},
  Phys. Rev. Lett. \textbf{75} (1995), no.~22, 4071--4074.

\end{thebibliography}

\end{document}